\pdfoutput=1
\PassOptionsToPackage{dvipsnames}{xcolor}
\documentclass[sigconf, 10pt, natbib=false]{acmart}

\copyrightyear{2019}
\acmYear{2019}
\setcopyright{acmcopyright}
\acmConference[SIGMOD'19]{2019 International Conference on Management of Data}{June 30 -- July 5, 2019}{Amsterdam, NL} 
\acmPrice{15.00}

\usepackage[
    backend=biber,
    sorting=none,
    maxbibnames=10,
    firstinits=true,
    abbreviate=true,
    dateabbrev=true,
    isbn=false,
    doi=false,
    style=ACM-Reference-Format
]{biblatex}
\renewcommand{\bibfont}{\bibliofont}
\AtEveryBibitem{
    \clearlist{publisher}
    \clearlist{location}
    \ifentrytype{misc}{}{\clearfield{url}\clearfield{urldate}}
}

\usepackage{ifthen}

\usepackage{color}

\long\def\eat#1{}

\usepackage[unicode]{hyperref}
\hypersetup{%
    hidelinks=true,
    bookmarksnumbered=true,
    bookmarksopen=true,
    bookmarksopenlevel=3,
    linktoc=page,
    plainpages=false,
    pdfstartview={XYZ null null 1},
    pdfpagemode={UseOutlines},
    pdfpagelayout={OneColumn}
}

\usepackage{caption}
\usepackage{subcaption}
\usepackage{booktabs}
\usepackage{tabularx}
\usepackage{multirow}
\usepackage{makecell}
\usepackage{tabu}
\usepackage{threeparttable}

\usepackage{tikz}
\usetikzlibrary{arrows,backgrounds,calc,decorations.pathreplacing,fit,matrix,patterns,positioning,shapes,shapes.multipart,decorations.markings}
\usepackage{pgfplots}
\makeatletter
\newenvironment{customlegend}[1][]{%
    \begingroup
    \csname pgfplots@init@cleared@structures\endcsname
    \pgfplotsset{#1}
    \def\addlegendimage{\csname pgfplots@addlegendimage\endcsname}
    \def\addlegendentry{\csname pgfplots@addlegendentry\endcsname}
}{%
    \csname pgfplots@createlegend\endcsname
    \endgroup
}%
\makeatother

\usepackage{paralist}
\usepackage[inline]{enumitem}
\setlist[itemize]{noitemsep,partopsep=0pt,topsep=0pt,leftmargin=1.5em}

\usepackage[ruled,vlined]{algorithm2e}
\SetKwProg{Fn}{Function}{}{}
\SetVlineSkip{0pt}
\SetInd{0.25em}{1em}
\SetAlCapNameFnt{\small}
\SetAlCapFnt{\small}

\usepackage{fancynum}
\usepackage{appendix}
\usepackage{thmtools,thm-restate}
\usepackage{balance}
\usepackage{etoolbox}
\usepackage{float}
\usepackage{placeins}
\usepackage{balance}
\usepackage{adjustbox}

\apptocmd\normalsize{%
    \setlength{\abovedisplayskip}{2pt}
    \setlength{\belowdisplayskip}{2pt}
    \setlength{\abovedisplayshortskip}{2pt}
    \setlength{\belowdisplayshortskip}{2pt}
}{}{}

\usepackage{listings}

\lstdefinelanguage{Solidity}{
    keywords=[1]{anonymous, assembly, assert, balance, break, call, callcode, case, catch, class, constant, continue, contract, debugger, default, delegatecall, delete, do, else, emit, event, export, external, false, finally, for, function, gas, if, implements, import, in, indexed, instanceof, interface, internal, is, length, library, log0, log1, log2, log3, log4, memory, modifier, new, payable, pragma, private, protected, public, pure, push, require, return, returns, revert, selfdestruct, send, storage, struct, suicide, super, switch, then, this, throw, transfer, true, try, typeof, using, value, view, while, with, addmod, ecrecover, keccak256, mulmod, ripemd160, sha256, sha3}, 
    keywordstyle=[1]\color{blue}\bfseries,
    keywords=[2]{address, bool, byte, bytes, bytes1, bytes2, bytes3, bytes4, bytes5, bytes6, bytes7, bytes8, bytes9, bytes10, bytes11, bytes12, bytes13, bytes14, bytes15, bytes16, bytes17, bytes18, bytes19, bytes20, bytes21, bytes22, bytes23, bytes24, bytes25, bytes26, bytes27, bytes28, bytes29, bytes30, bytes31, bytes32, enum, int, int8, int16, int24, int32, int40, int48, int56, int64, int72, int80, int88, int96, int104, int112, int120, int128, int136, int144, int152, int160, int168, int176, int184, int192, int200, int208, int216, int224, int232, int240, int248, int256, mapping, string, uint, uint8, uint16, uint24, uint32, uint40, uint48, uint56, uint64, uint72, uint80, uint88, uint96, uint104, uint112, uint120, uint128, uint136, uint144, uint152, uint160, uint168, uint176, uint184, uint192, uint200, uint208, uint216, uint224, uint232, uint240, uint248, uint256, var, void, ether, finney, szabo, wei, days, hours, minutes, seconds, weeks, years},	
    keywordstyle=[2]\color{teal}\bfseries,
    keywords=[3]{blockhash, coinbase, difficulty, gaslimit, number, timestamp, msg, data, gas, sender, sig, value, now, tx, gasprice, origin},	
    keywordstyle=[3]\color{violet}\bfseries,
    identifierstyle=\color{black},
    sensitive=false,
    comment=[l]{//},
    morecomment=[s]{/*}{*/},
    commentstyle=\color{gray}\ttfamily,
    stringstyle=\color{red}\ttfamily,
    morestring=[b]',
    morestring=[b]"
}

\theoremstyle{acmplain}

\theoremstyle{acmdefinition}
\newtheorem{definition}{Definition}[section]
\newtheorem{example}{Example}[section]

\usepackage{amssymb}
\let\emptyset\varnothing

\title{vChain: Enabling Verifiable Boolean Range Queries over Blockchain Databases}

\author{Cheng Xu}
\affiliation{%
    \institution{Hong Kong Baptist University}
}
\email{chengxu@comp.hkbu.edu.hk}

\author{Ce Zhang}
\affiliation{%
    \institution{Hong Kong Baptist University}
}
\email{cezhang@comp.hkbu.edu.hk}

\author{Jianliang Xu}
\affiliation{%
    \institution{Hong Kong Baptist University}
}
\email{xujl@comp.hkbu.edu.hk}

\begin{document}

\begin{abstract}
Blockchains have recently been under the spotlight due to the boom of cryptocurrencies and decentralized applications. There is an increasing demand for querying the data stored in a blockchain database. To ensure query integrity, the user can maintain the entire blockchain database and query the data locally. However, this approach is not economic, if not infeasible, because of the blockchain's huge data size and considerable maintenance costs. In this paper, we take the first step toward investigating the problem of verifiable query processing over blockchain databases. We propose a novel framework, called vChain, that alleviates the storage and computing costs of the user and employs verifiable queries to guarantee the results' integrity. To support verifiable Boolean range queries, we propose an accumulator-based authenticated data structure that enables dynamic aggregation over arbitrary query attributes. Two new indexes are further developed to aggregate intra-block and inter-block data records for efficient query verification. We also propose an inverted prefix tree structure to accelerate the processing of a large number of subscription queries simultaneously. Security analysis and empirical study validate the robustness and practicality of the proposed techniques.
\end{abstract}


\maketitle

\section{Introduction}\label{sec:introduction}

Owing to the success of cryptocurrencies such as Bitcoin~\cite{nakamoto2008bitcoin} and Ethereum~\cite{wood2014ethereum}, blockchain technology has been gaining overwhelming momentum in recent years. A blockchain is an append-only data structure that is distributively stored among peers in the network. Although peers in the network may not trust each other, a blockchain ensures data integrity from two aspects. First, powered by the hash chain technique,
data stored on a blockchain are immutable. Second, thanks to its consensus protocol, a blockchain guarantees that all peers maintain identical replicas of the data.
These cryptographically guaranteed security mechanisms, together with the decentralization and provenance properties of a blockchain, make blockchains a potential technology to revolutionize database systems~\cite{mohan2017tutorial,Dinh2017BLOCKBENCHAF,dinh2018untangling,Vo2018EDBT,wang2018forkbase}.

From the database perspective, a blockchain can be seen as a database storing a large collection of timestamped data records. With widespread adoption of blockchains for data-intensive applications such as finance, supply chains, and IP rights management, there is an increasing demand from users to query the data stored in a blockchain database. For example, in the Bitcoin network, users may want to find the transactions that satisfy a variety of \emph{range} selection predicates, such as \textit{``Transaction Fee $\geq$ \$50''} and \textit{``\$0.99M $\leq$ Total Output $\leq$ \$1.01M''}~\cite{Blockchair}. In a blockchain-based patent management system, users can use \emph{Boolean} operators to search for combinations of keywords, such as \textit{$\text{``Blockchain''} \land (\text{``Query''} \lor \text{``Search''})$}, in the patents' abstracts~\cite{Vaultitude}. Whereas many companies, including database giants IBM, Oracle, and SAP, as well as startups such as FlureeDB~\cite{FlureeDB}, BigchainDB~\cite{BigchainDB}, and SwarmDB~\cite{SwarmDB}, have devoted their efforts to developing blockchain database solutions to support SQL-like queries, all of them assume the existence of a trusted party who can faithfully execute user queries that are based on a materialized view of the blockchain database. However, such a trusted party may not always exist and the integrity of query results cannot be~guaranteed. Query processing with integrity assurance remains an unexplored issue in blockchain~research.

In a typical blockchain network~\cite{nakamoto2008bitcoin, wood2014ethereum},\footnote{For ease of exposition, in this paper we focus our discussion on public blockchains, but the proposed verifiable query techniques can be easily extended to private blockchains.} there are three types of nodes as shown in Fig.~\ref{fig:blockchainnet}: \emph{full node}, \emph{miner}, and \emph{light node}. A full node stores all the data in the blockchain, including block headers and data records. A miner is a full node with great computing power, responsible for constructing \emph{consensus proofs} (e.g., \emph{nonce} in the Bitcoin blockchain). A light node stores only block headers, which include the consensus proof and the cryptographic hashes of a block. Note that the data records are not stored in light nodes.

\begin{figure}[t]
    \centering
    \resizebox{.92\linewidth}{!}{\input{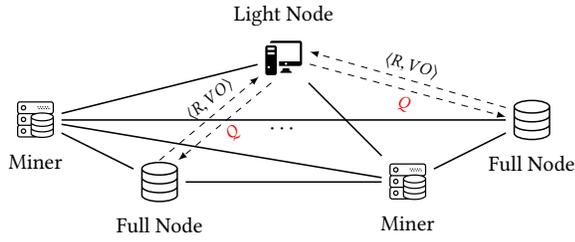}}
    \caption{A Blockchain Network}\label{fig:blockchainnet}
\end{figure}

To ensure the integrity of queries over a blockchain database, the query user could join the blockchain network as a full node. Then, the user can download and validate the entire database and process queries locally without compromising the query integrity. However, maintaining a full copy of the entire database might be too costly to an ordinary user, as it requires considerable storage, computing, and bandwidth resources.
For example, the minimum requirements of running a Bitcoin full node include 200GB of free disk space, an unmetered broadband connection with upload speeds of at least 50KB per second, and a running time of 6 hours a day~\cite{Bitcoin_fullnode}. To cater to query users with limited resources, especially mobile users, a more appealing alternative is to delegate the storage and query services to a powerful full node, while the query user only acts as a light node to receive results. Nevertheless, how to ensure the integrity of query results remains a challenge because full nodes are untrusted and that is an intrinsic assumption of the blockchain.

To address the aforementioned query integrity issue, in this paper, we propose a novel framework, called vChain, that employs \emph{verifiable query processing} to guarantee the results' integrity. More specifically, we augment each block with some additional \emph{authenticated data structure} (ADS), based on which an (untrusted) full node can construct and return a cryptographic proof, known as \emph{verification object} (VO), for users to verify the results of each query. The communication between a query user (light node) and a full node is illustrated in Fig.~\ref{fig:blockchainnet}, where $Q$ denotes a query request and $R$ denotes the result~set.

It is worth noting that this vChain framework is inspired by query authentication techniques studied for outsourced databases~\cite{Pang04, Li06,Yang09, Xu_SIGMOD15, Xu_SIGMOD18}. However, there are several key differences that render the conventional techniques inapplicable to blockchain databases. First, the conventional techniques rely on a data owner to sign the ADS using a private key. In contrast, in the blockchain network there is no data owner. Only the miners can append new data to the blockchain by constructing consensus proofs according to the consensus protocol. However, they cannot act as the data owner because they cannot hold the private key and sign the ADS. Second, a conventional ADS is built on a fixed dataset, and such an ADS cannot be efficiently adapted to a blockchain database in which the data are unbounded.
Third, in traditional outsourced databases, new ADSs can always be generated and appended, as needed, to support more queries involving different sets of attributes. However, that would be difficult due to the immutability of the blockchain, where a one-size-fits-all ADS is more desirable to support dynamic query attributes.

Clearly, the design of the ADS is a key issue of the vChain framework. To address this issue, this paper focuses on \emph{Boolean range queries}, which, as illustrated earlier, are commonly found in blockchain applications~\cite{Blockchair, Vaultitude}.
We propose a novel accumulator-based ADS scheme that enables dynamic aggregation over arbitrary query attributes, including both numerical attributes and set-valued attributes. This newly designed ADS is independent of the consensus protocol so that it is compatible with the current blockchain technology. On that basis, efficient verifiable query processing algorithms are developed. We also propose two authenticated indexing structures for intra-block data and inter-block data, respectively, to enable batch verification. To support large-scale subscription queries, we further propose a query indexing scheme that can group similar query requests. To summarize, our contributions made in this paper are as~follows:
\begin{itemize}
    \item To the best of our knowledge, this is the first work on verifiable query processing that leverages built-in ADSs to achieve query integrity for blockchain~databases.
    \item We propose a novel vChain framework, together with a new ADS scheme and two indexing structures that can aggregate intra-block and inter-block data records for efficient query processing and verification.
    \item We develop a new query index that can handle a large number of subscription queries simultaneously.
    \item  We conduct a security analysis as well as an empirical study to validate the proposed techniques. We also address the practical implementation issues.
\end{itemize}

The rest of the paper is organized as follows. Section~\ref{sec:related_work} reviews existing studies on blockchains and verifiable query processing. Section~\ref{sec:problem_definition} introduces the formal problem definition, followed by cryptographic primitives in Section~\ref{sec:prelimiaries}. Section~\ref{sec:basic_solution} presents our basic solution, which is then improved by two indexing structures devised in Section~\ref{sec:auth_index}. The verifiable subscription query is discussed in Section~\ref{sec:query_index}. The security analysis is presented in Section~\ref{sec:security_analysis}. Section~\ref{sec:experiment} presents the  experimental results.
Finally, we conclude our paper in Section~\ref{sec:conclusion}.

\begin{figure}[t]
    \centering
    \resizebox{.7\linewidth}{!}{\begin{tikzpicture}
    \tikzstyle{block node}=[
    rectangle split,
    rectangle split horizontal,
    rectangle split parts=4,
    rectangle split ignore empty parts,
    rectangle split part align=base,
    draw
    ]
    \tikzstyle{blockchain node}=[
    rectangle split,
    rectangle split horizontal,
    rectangle split parts=1,
    rectangle split part align=base,
    draw
    ]

    \tikzstyle{tree node}=[circle, draw, align=center, inner sep=0pt, text centered]
    \tikzstyle{tree}=[
    every node/.style={tree node},
    level/.style={level distance=0.6cm},
    level 1/.style={sibling distance=1.5cm},
    level 2/.style={sibling distance=0.8cm},
    ]

    \node (c0) {$\dots$};
    \node[blockchain node, fill=black!20, scale = 2, right=0.7cm of c0] (c1) { };
    \node[blockchain node, fill=black!20, scale = 2, right=0.7cm of c1] (c2) { };
    \node[blockchain node, fill=black!20, scale = 2, right=0.7cm of c2] (c3) { };
    \node [right=0.7cm of c3] (c4) {$\dots$};

    \draw[-latex] (c1.west) -- (c0.east);
    \draw[-latex] (c2.west) -- (c1.east);
    \draw[-latex] (c3.west) -- (c2.east);
    \draw[-latex] (c4.west) -- (c3.east);

    \node[block node, below=0.5cm of c2.south] (n) {
        \nodepart{one} PreBkHash
        \nodepart{two} TS
        \nodepart{three} ConsProof
        \nodepart{four} MerkleRoot
    };

    \path[tree]
    node[below=0.2cm of n.south] (root) {$H_{r}$}
    child {
            node (node1) {$H_{5}$}
            child { node (node3) {$H_{1}$} }
            child { node (node4) {$H_{2}$} }
        }
    child {
            node (node2) {$H_{6}$}
            child { node (node5) {$H_{3}$} }
            child { node (node6) {$H_{4}$} }
        };

    \draw[-latex] (root.east) -| (n.four south);

    \node[below=0.2cm of node3.south,draw,font=\footnotesize] (tx1) {$\textsf{Tx}_1$};
    \node[below=0.2cm of node4.south,draw,font=\footnotesize] (tx2) {$\textsf{Tx}_2$};
    \node[below=0.2cm of node5.south,draw,font=\footnotesize] (tx3) {$\textsf{Tx}_3$};
    \node[below=0.2cm of node6.south,draw,font=\footnotesize] (tx4) {$\textsf{Tx}_4$};
    \draw[-latex] (tx1.north) -- (node3.south);
    \draw[-latex] (tx2.north) -- (node4.south);
    \draw[-latex] (tx3.north) -- (node5.south);
    \draw[-latex] (tx4.north) -- (node6.south);

    \node[draw,fit=(n)(root)(tx1)(tx4)] (n-box) {};

    \node [left=0.7cm of n] (prev) {$\dots$};
    \node [right=0.7cm of n] (next) {$\dots$};
    \draw [-latex] (next) -- (n -| n-box.east);
    \draw [-latex] (n.west) -- (prev);

    \node[draw,dashed,rounded corners,fit=(c2)] (c2-box) {};
    \draw[dashed] (c2-box.south west) -- (n-box.north west);
    \draw[dashed] (c2-box.south east) -- (n-box.north east);
\end{tikzpicture}}
    \caption{Blockchain Structure}\label{fig:blockchain}
\end{figure}
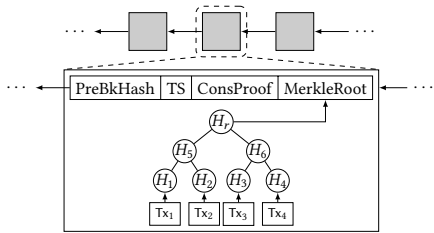

\section{Related Work}\label{sec:related_work}

In this section, we briefly review related studies and discuss relevant techniques.

\textbf{Blockchain.}
Since the introduction of Bitcoin cryptocurrency, blockchain technology has received considerable attention from both academia and industry~\cite{nakamoto2008bitcoin,wood2014ethereum,dinh2018untangling}. A blockchain, which is essentially a special form of \emph{Merkle Hash Tree} (MHT)~\cite{Merkle:1989ds}, is constructed as a sequence of blocks. As shown in Fig.~\ref{fig:blockchain}, each block stores a list of transaction records and an MHT built on top of them. The header of each block consists of four components:
\begin{enumerate*}[label=(\roman*)] 
    \item \textit{PreBkHash}, which is the hash of the previous block;
    \item \textit{TS}, which is the timestamp when the block was created;
    \item \textit{ConsProof}, which is constructed by the miners and guarantees the consensus of the block; and
    \item \textit{MerkleRoot}, which is the root hash of the MHT.
    \end{enumerate*}
The \textit{ConsProof} is usually computed based on the \textit{PreBkHash} and \textit{MerkleRoot}, and varies depending on the consensus protocol. In the widely used \emph{Proof of Work} (PoW) consensus protocol, the \textit{ConsProof} is a \textit{nonce} computed by the miners such that:
\begin{eqnarray*}
\textsf{hash}({PreBkHash}~|~TS~|~{MerkleRoot}~|~nonce)\leq Z
\end{eqnarray*}
where $Z$ corresponds to the mining difficulty. After a miner finds the nonce, it will pack the new block and broadcast it to the entire network. Other miners verify the transaction records and the nonce of the new block and, once verified, append it to the blockchain.

Significant effort has been made to address the various issues of blockchain systems, including system protocols~\cite{Ateniese2014CertifiedB,Garay2014TheBB}, consensus algorithms~\cite{eyal2016bitcoin,pirlea2018mechanising}, security~\cite{dong2017betrayal,camenisch2017practical}, storage~\cite{wang2018forkbase}, and performance benchmarking~\cite{Dinh2017BLOCKBENCHAF}.
Recently, major database vendors, including IBM~\cite{IBM}, Oracle~\cite{Oracle}, and SAP~\cite{SAP}, all have integrated blockchains with their database management systems, and they allow users to execute queries over blockchains through a database frontend. Besides, many startups such as FlureeDB~\cite{FlureeDB}, BigchainDB~\cite{BigchainDB}, and SwarmDB~\cite{SwarmDB}, have been developing blockchain-based database solutions for decentralized applications. However, they generally separate query processing from the underlying blockchain storage and count on trusted database servers for query integrity assurance. In contrast, our proposed vChain solution builds authenticated data structures into the blockchain structure, so that even untrusted servers can be enabled to offer integrity-assured query services.

\textbf{Verifiable Query Processing.} Verifiable query processing techniques have been extensively studied to ensure result integrity against an untrusted service provider (e.g., \cite{Pang04, Li06,Yang09, Xu_SIGMOD15, Xu_SIGMOD18, zhang2017vsql}). Most of the existing studies focus on outsourced databases and there are two typical approaches: supporting general queries using circuit-based verifiable computation (VC) techniques and supporting specific queries using an authenticated data structure (ADS).
The VC-based approach (e.g., SNARKs~\cite{parno2013pinocchio}) can support arbitrary computation tasks but at the expense of a very high and sometimes impractical overhead. Moreover, it entails an expensive preprocessing step as both the data and the query program need to be hard-coded into the proving key and the verification key. To remedy this issue, Ben-Sasson et al.~\cite{Ben-Sasson:2014:SNZ:2671225.2671275} have developed a variant of SNARKs in which the preprocessing step is only dependent on the upper-bound size of the database and query program. More recently, Zhang et al.~\cite{zhang2017vsql} have proposed a vSQL system, which utilizes an interactive protocol to support verifiable SQL queries. However, it is limited to relational databases with a fixed schema.

The ADS-based approach in comparison is generally more efficient as it tailors to specific queries. Our proposed solution belongs to this approach.
Two types of structures are commonly used to serve as an ADS: digital signature and MHT. Digital signatures authenticate the content of a digital message based on asymmetric cryptography. To support verifiable queries, it requires every data record to be signed and hence cannot scale up to large datasets~\cite{Pang04}. MHT, on the other hand, is built on a hierarchical tree~\cite{Merkle:1989ds}. Each entry in a leaf node is assigned a hash digest of a data record, and each entry in an internal node is assigned a digest derived from the child nodes. The data owner signs the root digest of MHT, which can be used to verify any subset of data records. MHT has been widely adapted to various index structures~\cite{Li06,Yang09,Xu_SIGMOD15}. More recently, there have been studies of verifiable queries on set-valued data~\cite{papamanthou2011optimal, Canetti14, Papadopoulos14, zhangexpressive, Xu_TKDE17}.

Another closely related line of research is on verifiable query processing for data streams~\cite{Schroder:2012io, Papamanthou:2013is,Schrder2015VeriStreamA, Papadopoulos2007CADSCA}. However, previous studies~\cite{Papamanthou:2013is,Schrder2015VeriStreamA} focus on \emph{one-time} queries to retrieve the latest version of streamed data. \cite{Papadopoulos2007CADSCA} requires the data owner to maintain an MHT for all data records and suffers from long query latency, which is not suitable for real-time streaming services. On the other hand, subscription queries over data streams have been investigated in~\cite{Chen:SIGMOD13,Thoma16, Yangkan17}. So far, no work has considered the integrity issue for subscription queries over blockchain databases.

\section{Problem Definition}\label{sec:problem_definition}

\begin{figure}[t]
    \begin{center}
        \resizebox{1\linewidth}{!}{\input{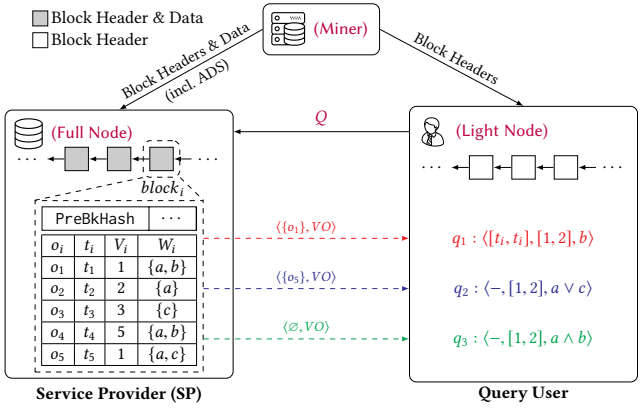}}
    \end{center}
    \caption{System Model of vChain}\label{fig:model}
\end{figure}

As was explained in Section 1, this paper proposes a novel vChain framework and studies verifiable query processing over blockchain databases. Fig.~\ref{fig:model} shows the system model of vChain, which involves three parties:
\begin{enumerate*}[label=(\roman*)] 
    \item \emph{miner},
    \item \emph{service provider} (SP), and
    \item \emph{query user}.
\end{enumerate*}
Both the miner and the SP are full nodes that maintain the entire blockchain database. The query user is a light node that keeps track of the block headers only. The miner is responsible for constructing the consensus proofs and appending new blocks to the blockchain. The SP provides query services to the lightweight user.

The data stored in the blockchain can be modeled as a sequence of blocks of temporal objects $\{o_1$, $o_2$, $\cdots$, $o_n\}$. Each object $o_i$ is represented by $\langle t_i, V_i, W_i \rangle$, where $t_i$ is the timestamp of the object, $V_i$ is a multi-dimensional vector that represents one or more numerical attributes, and $W_i$ is a set-valued attribute. To enable verifiable query processing, an \emph{authenticated data structure} (ADS) is constructed and embedded into each block by the miners (to be detailed in Sections~\ref{sec:basic_solution}-\ref{sec:query_index}). We consider two forms of Boolean range queries: (historical) time-window queries and subscription queries.

\textbf{Time-Window Queries.}
Users may wish to search the records appearing in a certain time period. In such a case, a time-window query can be issued. Specifically, a time-window query is in the form of $q = \langle [t_s, t_e], [\alpha, \beta], \Upsilon \rangle$, where $[t_s, t_e]$ is a temporal range selection predicate for the time period, $[\alpha, \beta]$ is a multi-dimensional range selection predicate for the numerical attributes, and $\Upsilon$ is a monotone Boolean function on the set-valued attribute. As a result, the SP returns all objects such that $\{ o_i = \langle t_i, V_i, W_i \rangle~|~~t_i \in [t_s, t_e] \land V_i \in [\alpha, \beta] \land \Upsilon(W_i) = 1 \}$.
For simplicity, we assume that $\Upsilon$ is in a \emph{conjunctive normal form} (CNF).

\begin{example}
    \itshape%
    In a Bitcoin transaction search service, each object $o_i$ corresponds to a coin transfer transaction. It consists of a transfer amount stored in $V_i$ and a set of sender/receiver addresses stored in $W_i$. A user may issue a query $q = \langle [\text{2018-05}, \text{2018-06}]$, $[10, +\infty]$, $\texttt{send:1FFYc} \land \texttt{receive:2DAAf} \rangle$ to find all of the transactions happening from May to June of 2018 with a transfer amount larger than 10 and being associated with the addresses ``\texttt{send:1FFYc}'' and ``\texttt{receive:2DAAf}''. 
\end{example}

\textbf{Subscription Queries.}
In addition to time-window queries, users can register their interests through subscription queries. Specifically, a subscription query is in the form of $q= \langle -$, $[\alpha, \beta]$, $\Upsilon \rangle$, where $[\alpha, \beta]$ and $\Upsilon$ are identical to the query conditions in time-window queries. In turn, the SP continuously returns all objects such that $\{ o_i = \langle t_i, V_i, W_i \rangle~|~V_i \in [\alpha, \beta] \land \Upsilon(W_i) = 1 \}$ until the query is deregistered.

\begin{example}
    \itshape%
    In a blockchain-based car rental system, each rental object $o_i$ consists of a rental price stored in $V_i$ and a set of textual keywords stored in $W_i$. A user may subscribe to a query $q = \langle -, [200, 250]$, $\text{``Sedan''} \land (\text{``Benz''} \lor \text{``BMW''}) \rangle$ to receive all rental messages that have a price within the range [200, 250] and contain the keywords ``Sedan'' and ``Benz'' or ``BMW''.
\end{example}
Additional examples of time-window queries and subscription queries can be found in Fig.~\ref{fig:model}.

\textbf{Threat Model.} We consider the SP, as an untrusted peer in the blockchain network, to be a potential adversary. Due to various issues such as program glitches, security vulnerabilities, and commercial interests, the SP may return tampered or incomplete query results, thereby violating the expected security of the blockchain. To address such a threat, we adopt verifiable query processing that enables the SP to prove the integrity of query results. Specifically, during query processing, the SP examines the ADS embedded in the blockchain and constructs a \emph{verification object} (VO) that includes the verification information of the results. The VO is returned to the user along with the results. Using the VO, the user can establish the \emph{soundness} and \emph{completeness} of the query results, under the following criteria:
\begin{itemize}
    \item \textbf{Soundness.} None of the objects returned as results have been tampered with and all of them satisfy the query conditions.
    \item \textbf{Completeness.} No valid result is missing regarding the query window or subscription period.
\end{itemize}
The above security notions will be formalized when we perform our security analysis in Section~\ref{sec:security_analysis}.

The main challenge in this model is how to design the ADS so that it can be easily accommodated in the blockchain structure while \emph{cost-effective} VOs (incurring small bandwidth overhead and fast verification time) can be efficiently constructed for both time-window queries and subscription queries. We address this challenge in the next few sections.

\section{Preliminaries}\label{sec:prelimiaries}

This section gives some preliminaries on cryptographic constructs that are needed in our algorithm design.

\textbf{Cryptographic Hash Function.}
A cryptographic hash function $\textsf{hash}(\cdot)$ accepts an arbitrary-length string as its input and returns a fixed-length bit string. It is collision resistant and difficult to find two different messages, $m_1$ and $m_2$, such that $\textsf{hash}(m_1) = \textsf{hash}(m_2)$. Classic cryptographic hash functions include the SHA-1, SHA-2, and SHA-3 families.

\textbf{Bilinear Pairing.} Let $\mathbb{G}$ and $\mathbb{H}$ be two cyclic multiplicative groups with the same prime order $p$. Let $g$ be the generator of $\mathbb{G}$. A bilinear mapping is a function $e: \mathbb{G} \times \mathbb{G} \rightarrow \mathbb{H}$ with the following properties:
\begin{itemize}[noitemsep,partopsep=0pt,topsep=0pt,leftmargin=2em]
    \item \textsf{Bilinearity}: If $u, v \in \mathbb{G}$ and $e(u,v)\in\mathbb{H}$, then $e(u^a, v^b) = {e(u, v)}^{ab}$ for any $u,v$.
    \item \textsf{Non-degeneracy}: $e(g, g) \neq 1$.
\end{itemize}
Bilinear pairing serves as a basic operation for the multiset accumulator as shown later in this paper.

\textbf{$q$-Strong Diffie-Hellman ($q$-SDH) Assumption}~\cite{Boneh04b}.
Let $pub = (p$, $\mathbb{G}$, $\mathbb{H}$, $e$, $g)$ be a bilinear pairing as described above. It states that for all polynomials $q$ and for all probabilistic polynomial-time adversaries \textsf{Adv},
\begin{align*}
    \Pr[ & s \gets \mathbb{Z}_p; \sigma = (pub, g^s, \cdots, g^{s^q});            \\
         & (c, h) \gets \textsf{Adv}(\sigma): h = {e(g, g)}^{1/(c+s)}] \approx 0
\end{align*}

\textbf{$q$-Diffie-Hellman Exponent ($q$-DHE) Assumption}~\cite{Camenisch:2009:ABB:1531954.1531989}.
Let $pub = (p$, $\mathbb{G}$, $g)$ as described above. It states that for all polynomials $q$ and for all probabilistic polynomial-time adversaries \textsf{Adv},
\begin{align*}
    \Pr[ & s \gets \mathbb{Z}_p; \sigma = (pub, g^s, \cdots, g^{s^{q-1}}, g^{s^{q+1}}, \cdots, g^{s^{2q-2}}); \\
         & h \gets \textsf{Adv}(\sigma): h = g^{s^q} ] \approx 0
\end{align*}

\textbf{Cryptographic Multiset Accumulator.}
A multiset is a generalization of a set in which elements are allowed to occur more than once. To represent them in a constant size, a cryptographic mulitset accumulator is a function $\textsf{acc}(\cdot)$, which maps a multiset to an element in some cyclic multiplicative group in a collision resistant fashion~\cite{Xu_TKDE17}.

One useful property of the accumulator is that it can be used to prove set disjoint. It consists of the following probabilistic polynomial-time algorithms:
\begin{itemize}
    \item \textsf{KeyGen}$(1^\lambda) \to (sk, pk)$: On input a security parameter $1^\lambda$, it generates a secret key $sk$ and a public key $pk$.
    \item \textsf{Setup}$(X, pk) \to \textsf{acc}(X)$: On input a multiset $X$ and the public key $pk$, it computes the accumulative value~$\textsf{acc}(X)$.
    \item \textsf{ProveDisjoint}$(X_1, X_2, pk) \to \pi$: On input two multisets $X_1$, $X_2$, where $X_1 \cap X_2 = \emptyset$,  and the public key $pk$, it outputs a proof $\pi$.
    \item \textsf{VerifyDisjoint}$(\textsf{acc}(X_1), \textsf{acc}(X_2), \pi, pk) \to \{0, 1\}$: On input the accumulative values $\textsf{acc}(X_1)$, $\textsf{acc}(X_2)$, a proof $\pi$, and the public key $pk$, it outputs 1 if and only if $X_1 \cap X_2 = \emptyset$.
\end{itemize}

More elaborated constructions of the accumulator and the set disjoint proof will be given in Section~\ref{sec:acc}.

\section{Basic Solution}\label{sec:basic_solution}

To enable verifiable queries in our vChain framework, a naive scheme is to construct a traditional MHT as the ADS for each block and apply the conventional MHT-based authentication methods. However, this naive scheme has three major drawbacks. First, an MHT  supports only the query keys on which the Merkle tree is built. To support queries involving an arbitrary set of attributes, an exponential number of MHTs need to be constructed for each block. Second, MHTs do not work with set-valued attributes. Third, MHTs of different blocks cannot be aggregated efficiently, making it incapable of leveraging inter-block optimization techniques. To overcome these drawbacks, in this section we propose novel authentication techniques based on a new accumulator-based ADS scheme, which transforms numerical attributes into set-valued attributes and enables dynamic aggregation over arbitrary query attributes.

In the following, we start by considering a single object and focusing on the Boolean time-window query for ease of illustration (Sections~\ref{sec:ads_query} and \ref{sec:acc}). We then extend it to the range query condition (Section~\ref{sec:extendrange}). We discuss the batch query processing and verification for multiple objects in Section~\ref{sec:auth_index}. The subscription query is elaborated in Section~\ref{sec:query_index}.

\vspace{-2ex}
\subsection{ADS Generation and Query Processing}\label{sec:ads_query}

For simplicity, this section considers the Boolean query condition on the set-valued attribute $W_i$ only.
We assume that each block stores a single object $o_i = \langle t_i, W_i \rangle$ and use \textit{ObjectHash} to denote \textit{MerkleRoot} in the original block structure (Fig.~\ref{fig:blockchain}).

\begin{figure}[t]
    \centering
    \resizebox{.9\linewidth}{!}{\begin{tikzpicture}
    \tikzstyle{block node}=[
    rectangle split,
    rectangle split horizontal,
    rectangle split parts=5,
    rectangle split ignore empty parts,
    rectangle split part align=base,
    draw
    ]

    \node (prev) {$\dots$};

    \node[block node, right=0.8 of prev] (n) {
        \nodepart{one} PreBkHash
        \nodepart{two} TS
        \nodepart{three} ConsProof
        \nodepart{four} ObjectHash
        \nodepart{five} AttDigest
    };
    \node[scale=1.2,above=0.1 of n] (n-label) {block$_{i}$};

    \node[draw, below=0.3 of n.four south] (o) {$o_{i}$};
    \draw[-latex] (o) -- (n.four south);

    \node[draw,fit=(n)(o)] (n-box) {};

    \begin{pgfonlayer}{background}
        \fill[color=black!20] (n.four split north) rectangle (n.south east);
    \end{pgfonlayer}

    \node[right=0.8 of n] (next) {$\dots$};

    \draw[-latex] (n.west) -- (prev.east);
    \draw[-latex] (next.west) -- (n-box.east |- n.east);
\end{tikzpicture}}
    \caption{Extended Block Structure}\label{fig:block}
\end{figure}
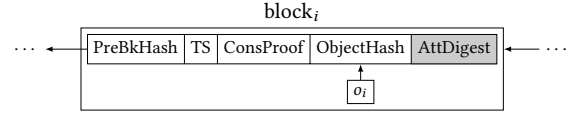

\textbf{ADS Generation.}
Recall that in the proposed vChain framework, an ADS is generated for each block during the mining process. It can be used by the SP to construct a verification object (VO) for each query. To this end, we extend the original block structure by adding an extra field, named \textit{AttDigest}, as shown by the shaded part in Fig.~\ref{fig:block}. Thus, the block header consists of \textit{PreBkHash}, \textit{TS}, \textit{ConsProof}, \textit{ObjectHash}, and \textit{AttDigest}.

To serve as the ADS, \textit{AttDigest} should have three desired properties. First, \textit{AttDigest} should be able to summarize an object's attribute $W_i$ in a way that it can be used to prove whether or not the object matches a query condition. In case of a mismatch, we can just return this digest instead of the whole object. Second, \textit{AttDigest} should be in a constant size regardless of the number of elements in $W_i$. Third, \textit{AttDigest} should be aggregatable to support batch verification of multiple objects within a block or even across blocks (Section~\ref{sec:auth_index}). As such, we propose to use  \textit{multiset accumulator} as \textit{AttDigest}:
\begin{eqnarray} \nonumber
AttDigest_i = \textsf{acc}(W_i) = \textsf{Setup}(W_i, pk)
\end{eqnarray}
While its supported functionalities, including \textsf{ProveDisjoint}($\cdot$) and \textsf{VerifyDisjoint}($\cdot$), have been described in Section~\ref{sec:prelimiaries}, for better readability, we defer detailed constructions to Section~\ref{sec:acc}.

\textbf{Verifiable Query Processing.} Given a Boolean query condition and a data object, there are only two possible outcomes: match or mismatch. The soundness of the first case can be easily verified by returning the object as a result, since its integrity can be authenticated by the \textit{ObjectHash} stored in the block header, which is available to the query user on a light node (recall Fig.~\ref{fig:model}). The challenge lies in how to effectively verify the second case by using \textit{AttDigest}.  As CNF is a Boolean function expressed in a list of AND of OR operators, we can view the Boolean function in CNF as a list of sets. For example, a query condition $\text{``Sedan''} \land (\text{``Benz''} \lor \text{``BMW''})$ is equivalent to two sets: $\{\text{``Sedan''}\}$ and $\{\text{``Benz''}, \text{``BMW''}\}$. Consider a mismatching object $o_i: \{\text{``Van''}, \text{``Benz''}\}$. It is easy to observe that there exists an equivalence set (i.e., $\{\text{``Sedan''}\}$) such that its intersection with the object's attribute is empty. Thus, we can apply \textsf{ProveDisjoint}(\{\text{``Van''}, \text{``Benz''}\}, $\{\text{``Sedan''}\}, pk$) to generate a disjoint proof $\pi$ as the VO for the mismatching object. Accordingly, the user can retrieve $\textit{AttDigest}_i=\textsf{acc}(\{\text{``Van''}, \text{``Benz''}\})$ from the block header and use $\textsf{VerifyDisjoint}(\textit{AttDigest}_i, \textsf{acc}(\{\text{``Sedan''}\})$, $\pi$, $pk)$ to verify the mismatch.
The whole process is detailed in Algorithm~\ref{alg:single-obj-query}.

\begin{algorithm}[t]
    \caption{Verifiable Query on a Single Object}\label{alg:single-obj-query}
    \small
    \SetKwBlock{DO}{ADS Generation (by the miner)}{}
    \SetKwBlock{SP}{VO Construction (by the SP)}{}
    \SetKwFunction{VOConstruction}{VOConstruction}
    \SetKwBlock{User}{Result Verification (by the user)}{}
    \DO{%
        \For{each object $o_i = \langle t_i, W_i \rangle$}{%
         	$\textit{AttDigest}_i \gets \textsf{acc}(W_i)$\;
            Write $\langle \textsf{hash}(o_i), \textit{AttDigest}_i \rangle$ to the block\;
        }
    }
    \SP{%
        \Fn{\VOConstruction{$o_i,q$}}{
            \KwIn{Object $o_i$, Query condition $q = \langle \Upsilon \rangle$}
            \lIf{$o_i$ matches $q$}{Send $o_i$ to the user}
            \Else{
                Interpret $\Upsilon$ as a list of sets $\{\Upsilon_1, \cdots, \Upsilon_\ell\}$, s.t.\ $\Upsilon = \land_{\Upsilon_i \in \{\Upsilon_1, \cdots, \Upsilon_\ell\}} (\lor_{x \in \Upsilon_i} x)$\;
                Find $\Upsilon_i$ such that $\Upsilon_i \in \{\Upsilon_1, \cdots, \Upsilon_\ell\} \land \Upsilon_i \cap W_i = \emptyset$\;
                $\pi \gets \textsf{ProveDisjoint}(W_i, \Upsilon_i, pk)$\;
                Send $\langle \pi, \Upsilon_i \rangle$ to the user\;
            }
        }
    }
    \User{%
    \eIf{$o_i$ matches $q$}{%
        Check $o_i$ w.r.t.\ $\textsf{hash}(o_i)$ from the block header\;
        Check whether $o_i$ matches $q$\;
    }{
        Read $\textit{AttDigest}_i$ from the block header\;
        Run $\textsf{VerifyDisjoint}(\textit{AttDigest}_i, \textsf{acc}(\Upsilon_i), \pi, pk)$\;
        }
    }
\end{algorithm}

It is straightforward to extend the above algorithm to support time-window queries. The process basically finds the corresponding blocks whose timestamp is within the query window and invokes Algorithm~\ref{alg:single-obj-query} repeatedly for each object in these selected blocks. For example, suppose the query Boolean function is $\text{``Sedan''} \land (\text{``Benz''} \lor \text{``BMW''})$. The list of objects that are within the time window includes $o_1$: $\{\text{``Sedan''}$, $\text{``Benz''}\}$, $o_2$: $\{\text{``Sedan''}$, $\text{``Audi''}\}$, $o_3$: $\{\text{``Van''}$, $\text{``Benz''}\}$, and $o_4$: $\{\text{``Van''}$, $\text{``BMW''}\}$. The SP can apply $\textsf{ProveDisjoint}(\cdot)$ for $o_2$, $o_3$, and $o_4$ to prove that they do not match the condition $\text{``Benz''} \lor \text{``BMW''}$, $\text{``Sedan''}$, and $\text{``Sedan''}$, respectively. As for $o_1$, the SP will return the object directly since it is a~match.

\subsection{Constructions of Multiset Accumulator}\label{sec:acc}
We now discuss two possible constructions of the multiset accumulator used as \textit{AttDigest} in Section~\ref{sec:ads_query}. Each construction has its own advantage and disadvantage, and is suitable to different application scenarios as we will see in Section~\ref{sec:experiment}.

\subsubsection{Construction 1}

We first present a construction proposed in~\cite{papamanthou2011optimal}, which is based on bilinear pairing and $q$-SDH assumption. It consists of the following algorithms.

\textsf{KeyGen}$(1^\lambda) \to (sk, pk)$:
Let $(p$, $\mathbb{G}$, $\mathbb{H}$, $e$, $g)$ be a bilinear pairing. Choose a random value $s \gets \mathbb{Z}_p$.
The secret key is $sk = (s)$ and the public key is $pk = (g, g^s, g^{s^2}, \cdots, g^{s^q})$.

\textsf{Setup}$(X, pk) \to \textsf{acc}(X)$:
The accumulative value for a multiset $X = \{ x_1, \cdots, x_n \}$ is $\textsf{acc}(X) = g^{P(X)} = g^{\prod_{x_i \in X}{(x_i + s)}}$. Owing to the property of the polynomial interpolation, it can be computed without knowing the secret key.

\textsf{ProveDisjoint}$(X_1, X_2, pk) \to \pi$:
According to the extended Euclidean algorithm, if $X_1 \cap X_2 = \emptyset$, there exist two polynomials $Q_1, Q_2$ such that $P(X_1)Q_1 + P(X_2)Q_2 = 1$. As such, if $X_1 \cap X_2 = \emptyset$, the proof can be computed as $\pi = (F_1^*, F_2^*) = (g^{Q_1}, g^{Q_2})$.

\textsf{VerifyDisjoint}$(\textsf{acc}(X_1), \textsf{acc}(X_2), \pi, pk) \to \{0, 1\}$:
To verify the proof, the verifier interprets $\pi$ as $(F_1^*, F_2^*)$. The proof is valid if and only if the following constraint holds:
\begin{align*}
    e(\textsf{acc}(X_1), F_1^*) \cdot e(\textsf{acc}(X_2), F_2^*) \stackrel{?}{=} e(g, g).
\end{align*}

\subsubsection{Construction 2}
\label{sec:construction2}
Inspired by~\cite{zhangexpressive}, the second construction is proposed to introduce two additional \textsf{Sum}($\cdot$) and \textsf{ProofSum}($\cdot$) primitives, which allow the aggregation of multiple accumulative values or set disjoint proofs. It is based on bilinear pairing and $q$-DHE assumption and consists of the following algorithms.

\textsf{KeyGen}$(1^\lambda) \to (sk, pk)$:
Let  $(p$, $\mathbb{G}$, $\mathbb{H}$, $e$, $g)$ be a bilinear pairing. Choose a random value $s \gets \mathbb{Z}_p$.
The secret key is $sk = (s)$ and the public key is $pk = (g$, $g^s$, $g^{s^2}$, $\cdots$, $g^{s^{q-1}}$, $g^{s^{q+1}}$, $\cdots$, $g^{s^{2q-2}})$.

\textsf{Setup}$(X, pk) \to \textsf{acc}(X)$:
The accumulative value for a multiset $X = \{ x_1, \cdots, x_n \}$ is $\textsf{acc}(X) = (d_A(X), d_B(X))$, where $d_A(X) = g^{A(X)} = g^{\sum_{x_i \in X}{s^{x_i}}}$ and $d_B(X) = g^{B(X)} = g^{\sum_{x_i \in X}{s^{q-x_i}}}$. Similar to the first construction, it can also be computed without knowing the secret key using the polynomial interpolation.

\textsf{ProveDisjoint}$(X_1, X_2, pk) \to \pi$:
It is easy to see that if $X_1 \cap X_2 = \emptyset$, then $C s^q \notin A(X_1) B(X_2)$, where $C$ is a non-zero constant. As such, if $X_1 \cap X_2 = \emptyset$, the proof can be computed as $\pi = g^{A(X_1)B(X_2)} = g^{(\sum_{x_{i} \in X_1} s^{x_{i}})(\sum_{x_{j} \in X_2 } s^{q-x_{j}})}$.

\textsf{VerifyDisjoint}$(\textsf{acc}(X_1), \textsf{acc}(X_2), \pi, pk) \to \{0, 1\}$:
To verify the proof, the verifier interprets $\textsf{acc}(X_1)$ and $\textsf{acc}(X_2)$ as $(d_A(X_1), d_B(X_1))$ and $(d_A(X_2), d_B(X_2))$, respectively. The proof is valid if and only if the following constraint holds:
\begin{align*}
    e(d_A(X_1), d_B(X_2)) \stackrel{?}{=} e(\pi, g).
\end{align*}

\textsf{Sum}$(\textsf{acc}(X_1), \textsf{acc}(X_2), \cdots, \textsf{acc}(X_n)) \to \textsf{acc}(\sum_i^n X_i)$:
On input of multiple accumulative values $\textsf{acc}(X_1), \cdots, \textsf{acc}(X_n)$, it outputs the accumulative value for the multiset $\sum_i^n X_i$. $\textsf{acc}(\sum_i^n X_i) = (d_A(\sum_i^n X_i), d_B(\sum_i^n X_i))$, where $d_A(\sum_i^n X_i) = \prod_i^n d_A(X_i)$ and $d_B(\sum_i^n X_i) = \prod_i^n d_B(X_i)$.

\textsf{ProofSum}$(\langle \pi_1, X_{\pi_1,1}, X_{\pi_1,2} \rangle$, $\cdots$, $\langle \pi_n, X_{\pi_n,1} , X_{\pi_n,2} \rangle) \to \pi'$:
On input of multiple set disjoint proofs $\pi_1 = \textsf{ProveDisjoint}(\allowbreak X_{\pi_1,1}$, $X_{\pi_1,2}$, $pk)$, $\cdots$, $\pi_n = \textsf{ProveDisjoint}(X_{\pi_n,1}$, $X_{\pi_n,2}, pk)$, it outputs an aggregate proof $\pi'=\sum_{i=1}^n \pi_i$, if and only if $X_{\pi_1,2}=X_{\pi_2,2}=\cdots = X_{\pi_n,2}$.

Compared with Construction 1, Construction 2 supports the aggregation of multiple accumulative values or set disjoint proofs, which can be used by the online batch verification method in Section~\ref{sec:online_batch}. However, it incurs a much larger key size. In particular, the public key size in Construction 1 is linear to the largest multiset size, whereas in Construction~2, the public key size is linear to the largest possible value of the attributes in the system. In real-life applications, the common practice is to use a cryptographic hash function to encode each attribute value into an integer number, which is then accepted by the accumulator. Since the value returned by a typical hash function is in several hundreds of bits, it is costly to generate and publish the public key with such a scale in advance. To remedy this issue, we may introduce a trusted oracle, which owns the secret key and is responsible to answer requests of the public key. Such an oracle can be acted by a trusted third party or be implemented utilizing secure hardware like SGX\@.

\subsection{Extension to Range Queries}
\label{sec:extendrange}
The previous sections mainly consider the Boolean queries on the set-valued attribute $W_i$. In many scenarios, the user may also apply range conditions on the numerical attributes~$V_i$. To tackle this problem, we propose a method that transforms numerical attributes into set-valued attributes. Then, a range query can be mapped to a Boolean query accordingly.

The idea goes as follows. First, we represent each numerical value in the binary format. Next, we transform a numerical value into a set of binary prefix elements (denoted as function \textsf{trans}($\cdot$)).
For example, a value $4$ can be represented in the binary format $100$. Thus, it can be transformed into a prefix set, i.e., $\textsf{trans}(4) = \{1*, 10*, 100\}$, where $*$ denotes the wildcard matching operator.
Similarly for a numerical vector, we can apply the above procedure for each dimension. For example, a vector $(4, 2)$ has the binary format $(100, 010)$. Thus, its transformed prefix set is $\{{1*}_1, {10*}_1, {100}_1, {0*}_2, {01*}_2, {010}_2\}$. Note that here each element has a subscript notation (i.e.,\ $_1$ and $_2$), which is used to distinguish the binary values in the different dimensions of the vector.

Next, we transform a range query condition into a monotone Boolean function, by using a binary tree built over the entire binary space  (e.g., Fig.~\ref{fig:transform} shows a tree of single-dimension space [0,7]). Specifically, for a single-dimension range $[\alpha, \beta]$, we first represent $\alpha$ and $\beta$ in its binary format. Next, we view $\alpha$ and $\beta$ as two leaf nodes in the tree.
Finally, we find the minimum set of  tree nodes to exactly cover the whole range $[\alpha, \beta]$. The transformed Boolean function is a function concatenating each element in the set using OR ($\lor$) semantic. For example, for a query range $[0, 6]$, we can find its transformed Boolean function as $0* \lor 10* \lor 110$ (see the gray nodes in Fig.~\ref{fig:transform}). As discussed in Section~\ref{sec:ads_query}, the equivalence set of this Boolean function is $\{0*, 10*, 110\}$.
Similarly, in the case of a multi-dimensional range, the transformed Boolean function is the one concatenating the partial Boolean function for each dimension using AND ($\land$) semantic. For example, a query range $[(0, 3), (6, 4)]$ can be transformed to $({0*}_1 \lor {10*}_1 \lor {110}_1) \land ({011}_2 \lor {100}_2)$, with equivalence sets of $\{0*_1, 10*_1, 110_1\}$ and $\{011_2, 100_2\}$.

\begin{figure}[t]
    \centering
    \resizebox{0.67\linewidth}{!}{\begin{tikzpicture}
    \tikzstyle{tree node}=[circle, draw, align=center, inner sep=0pt, text centered]
    \tikzstyle{tree}=[
    every node/.style={minimum size=6mm,tree node},
    level/.style={level distance=0.8cm},
    level 1/.style={sibling distance=3.5cm},
    level 2/.style={sibling distance=1.8cm},
    level 3/.style={sibling distance=0.8cm},
    ]
    \path[tree]
    node[] (root) {$*$}
    child {
        node[fill=black!20] (n1) {$0*$}
        child { node (n3) {$00*$}
            child{
                node(n7) {$000$}}
            child{	node(n8) {$001$}}}
        child { node (n4) {$01*$}
            child{
                node(n9){$010$}}
            child{	node(n10){$011$}}
        }
    }
    child {
        node (n2) {$1*$}
        child { node[fill=black!20] (n5) {$10*$}
            child{	node(n11){$100$}}
            child{	node(n12){$101$}}
        }
        child { node (n6) {$11*$}
            child{	node[fill=black!20](n13){$110$}}
            child{	node(n14){$111$}}
        }
    };
    \draw[|-|, red] ($(n7.south)+(-4mm,-2mm)$) to ($(n13.south)+(4mm,-2mm)$);
    \node[below=2.6 of root]{$q=[0,6]$};
\end{tikzpicture}}
    \caption{Example of Transformation}\label{fig:transform}
\end{figure}
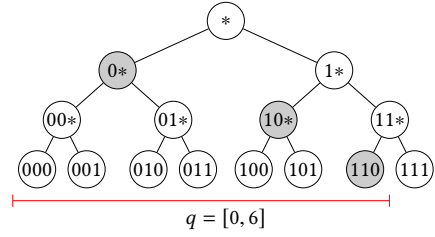

With the above transformations, a query of whether or not a numerical value $v_i$ is in a range $[\alpha, \beta]$ becomes a Boolean query of $v_i$'s transformed prefix set against $[\alpha, \beta]$'s equivalence sets. In the above examples, $4 \in [0, 6]$ since  $\{1*, 10*, 100\} \cap \{0*, 10*, 110\} = \{10*\} \neq \emptyset$; $(4,2) \notin [(0$, $3)$, $(6$, $4)]$ since there exists some equivalence set $\{011_2, 100_2\}$ such that $\{011_2, 100_2\} \cap \{{1*}_1, {10*}_1,$ ${100}_1, {0*}_2, {01*}_2, {010}_2\} = \emptyset$.

Thanks to the data transformation technique, in the sequel we unify the two types of query conditions into a uniform Boolean query condition on the set-valued attribute. More specifically, for each data object $\langle t_i, V_i, W_i \rangle$, it is transformed into a tuple $\langle t_i, W_i' \rangle$, where $W_i' = \textsf{trans}(V_i) + W_i$; and a query $q = \langle [t_s, t_e], [\alpha, \beta], \Upsilon \rangle$ is transformed into $\langle [t_s, t_e], \Upsilon ' \rangle$, where $\Upsilon ' = \textsf{trans}([\alpha, \beta]) \land \Upsilon$. As such, the query results are $\{ o_i = \langle t_i, W_i' \rangle~|~t_i \in [t_s, t_e] \land \Upsilon'(W_i') = 1 \}$.

\vspace{-2ex}
\section{Batch Verification}\label{sec:auth_index}

In this section, we discuss how to boost the query performance via batch verification. We first  introduce two authenticated indexing structures, namely intra-block index (Section~\ref{sec:intra_index}) and inter-block index (Section~\ref{sec:inter_index}), followed by an online batch verification method (Section~\ref{sec:online_batch}).  All these techniques allow the SP to prove mismatching objects in a~batch.

\subsection{Intra-Block Index}\label{sec:intra_index}

In the previous discussion, we assume that each block stores only one object for simplicity. In general, each block often stores multiple objects. Naively, we can apply the single-object algorithm repeatedly for each object to ensure query integrity, which however incurs a verification complexity linear to the number of objects. Further, it can be observed that if two objects share some common attribute value, they may mismatch some queries due to the same partial query condition. Therefore, to reduce the proofing and verification overhead, we propose an intra-block index which can aggregate multiple objects and improve performance.

\begin{figure}[t]
    \centering
    \resizebox{0.95\linewidth}{!}{\begin{tikzpicture}
    \tikzstyle{block node}=[
    rectangle split,
    rectangle split horizontal,
    rectangle split parts=4,
    rectangle split ignore empty parts,
    rectangle split part align=base,
    draw
    ]
    \tikzstyle{tree node}=[circle, draw, align=center, inner sep=0pt, text centered]
    \tikzstyle{tree}=[
    every node/.style={tree node},
    level/.style={level distance=0.5cm},
    level 1/.style={sibling distance=2cm},
    level 2/.style={sibling distance=1cm},
    ]

    \node[block node] (block) {
        \nodepart{one} PreBkHash
        \nodepart{two} TS
        \nodepart{three} ConsProof
        \nodepart{four} MerkleRoot
    };

    \path[tree]
    node[below=0.3cm of block.four south] (root) {$N_r$}
    child {
        node (n1) {$N_5$}
        child { node (n3) {$N_1$} }
        child { node[fill=black!20] (n4) {$N_2$} }
    }
    child {
        node (n2)[fill=black!20] {$N_6$}
        child { node (n5) {$N_3$} }
        child { node (n6) {$N_4$} }
    };

    \draw[-latex] (root.north) -- (block.four south);

    \node[block node, scale=0.9, left=0.9cm of root] (root-data) {
        \nodepart{one} \textit{hash}$_r$
        \nodepart{two} $W_r$
        \nodepart{three}\textit{AttDigest}$_r$
    };
    \draw[dashed,-latex] (root-data) -- (root);

    \node[block node, scale=0.9, left=0.5cm of n3] (n3-data) {
        \nodepart{one} \textit{hash}$_1$
        \nodepart{two} $W_1$
        \nodepart{three} \textit{AttDigest}$_1$
    };
    \draw[dashed,-latex] (n3-data) -- (n3);

    \node[draw, scale=0.9, below=0.3cm of n3-data.one south, fill=black!20] (o1) {$o_1$};
    \draw[-latex] (o1) -- (n3-data.one south);

    \node[draw,fit=(block)(n6)(root-data)(n3-data)(o1)] (box) {};
    \node[scale=1.2,above=0.05cm of box] (block-label) {block$_i$};

    \node[left=1cm of block] (prev) {$\dots$};
    \node[right=1.5cm of block] (next) {$\dots$};
    \draw[-latex] (block.west) -- (prev.east);
    \draw[-latex] (next.west) -- (box.east |- block.east);
    \node[below=0.1cm of box] (db) {
        \begin{tabu}{|c|c|l|}
            \hline
            \rowfont{\bfseries}
            Node & Object  &
            Set Attributes\\
            \hline
            $N_1$  & $o_1$  & $W_1=\{\text{``Sedan''}, \text{``Benz''}\}$ \\
            \hline
            $N_2$  & $o_2$  & $W_2=\{\text{``Sedan''}, \text{``Audi''}\}$ \\
            \hline
            $N_3$  & $o_3$  & $W_3=\{\text{``Van''}, \text{``Benz''}\}$ \\
            \hline
            $N_4$  & $o_4$  & $W_4=\{\text{``Van''}, \text{``BMW''}\}$ \\
            \hline
        \end{tabu}
    };
\end{tikzpicture}}
    \caption{Intra-Block Index}\label{fig:intra_index}
\end{figure}
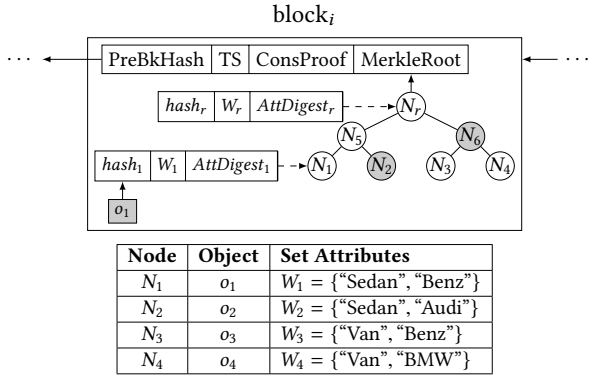

Fig.~\ref{fig:intra_index} shows a block of the blockchain with the intra-block index. It organizes the \textit{ObjectHash} and \textit{AttDigest} of each object into a binary Merkle tree. The block header consists of the following components: \textit{PreBkHash}, \textit{TS}, \textit{ConsProof}, and \textit{MerkleRoot}, where \textsf{MerkleRoot} is the root hash of the binary Merkle tree. Each tree node has three fields: child hash (denoted by \textit{hash}$_i$, and is used to form the MHT), attribute multiset (denoted by $W_i$), and attribute multiset's accumulative value (denoted by \textit{AttDigest}$_i$). They are computed from the child nodes as follows.

\begin{definition}[Intra-Block Index Non-Leaf Node]
    Let $\textsf{hash}(\cdot)$ be a cryptographic hash function, `$|$' be the string concatenation operator, $\textsf{acc}(\cdot)$ be the multiset accumulator, $n_l$ and $n_r$ be the left and right children of node $n$, respectively. The fields for a non-leaf node $n$ are defined as:
    \begin{itemize}
  \item     $ W_n = W_{n_l} \cup W_{n_r} $
  \item      $ {AttDigest}_n  = \textsf{acc}(W_n) $
  \item      $ {hash}_n       = \textsf{hash}(\textsf{hash}({hash}_{n_l}~|~{hash}_{n_r})~|~{AttDigest}_n ) $
    \end{itemize}
\end{definition}

\begin{definition}[Intra-Block Index Leaf Node]
    The fields for a leaf node are identical to those for the underlying object.
\end{definition}

When building the intra-block index, we want to achieve the maximum proofing efficiency. That is, we aim to maximize the chance of pruning the mismatching objects together during query processing. On the one hand, this means that we should find a clustering strategy such that given a user's query, the chance that a node mismatches the query is maximum. In another words, we strive to maximize the similarity of the objects under each node. On the other hand, a balanced tree is preferred since it can improve the query efficiency. Thus, we propose that the intra-block index is built based on the block's data objects in a bottom-up fashion (by the blockchain miners). First, each data object in the block is assigned to a leaf node. Next, the leaf nodes which yield the maximum Jaccard  similarity $\frac{|W_{n_l} \cap W_{n_r}|}{|W_{n_l} \cup W_{n_r}|}$ are iteratively merged. The two merged tree nodes are used to create a new non-leaf node in the upper level. This process is repeated in each level until the root node is created. Finally, the \textit{MerkleRoot} assigned by $hash_r$ is written as one of the components of the block header. Algorithm~\ref{alg:intraconstruct} shows the procedure in~detail.

\begin{algorithm}[t]
    \caption{\small Intra-Index Construction (by the miners)}\label{alg:intraconstruct}
    \small

    \SetKwFunction{Intra}{BuildIntraIndex}
    \Fn{\Intra{$nodes$}}{%
    \KwIn{list of leaf nodes $nodes$}
    	\While{$nodes.len>1$}{
    		$newnodes\gets [ ]$\;
    		\While{$nodes.len>1$}{
    		$n_l \gets \underset{n}{\operatorname{arg\,max}}\ |W_{n}|$;
    			$nodes$.delete($n_l$)\;
    		$n_r \gets \underset{n}{\operatorname{arg\,max}}\ \frac{|W_{n_l} \cap W_n|}{|W_{n_l} \cup W_n|}$;
    			$nodes$.delete($n_r$)\;
    			$W_{n}\gets W_{n_{l}}\cup W_{n_{r}}$\;
    			${AttDigest}_{n}\gets \textsf{acc}(node.W)$\;
    			${h}_{n}\gets\textsf{hash}(\textsf{hash}({hash}_{n_{l}}|{hash}_{n_{r}})|{AttDigest}_{n})$\;
    			$newnodes$.add($\langle n_l, n_r, {h}_n, W_n, {AttDigest}_n \rangle$)\;
    		}
    		$nodes\gets newnodes+nodes$\;
    	}
    $\textsf{root} \gets nodes[0]$\;
    $MerkleRoot\gets hash_r$\;
    }
\end{algorithm}
With the above intra-block index, the SP can process a query as a tree search. Starting from the root node, if the attribute multiset of the current node fulfills the query condition, its subtree will be further explored. Also, the corresponding \textit{AttDigest} is added to the VO, which will be used to reconstruct the \textit{MerkleRoot} during result verification. On the other hand, if the multiset does not satisfy the query condition, it means that all the underlying objects are mismatches. In this case, the SP will invoke $\textsf{ProveDisjoint}(\cdot)$ with the corresponding $\textit{AttDigest}$ to generate a mismatch proof. Upon reaching a leaf node, the object whose multiset satisfies the query condition is a matching object and will be returned as a query result. Algorithm~\ref{alg:intra-index-query} shows the VO construction using the intra-block index.

For illustration, we use the same set of objects as discussed in Section~\ref{sec:ads_query}. The intra-block index is shown in Fig.~\ref{fig:intra_index}. The Boolean query from the user is $\text{``Sedan''} \land (\text{``Benz''} \lor \text{``BMW''})$.
The query process simply traverses the index from the root node to the leaf nodes. The query result is $\{o_1\}$. The VO returned by the SP includes $\{\langle{AttDigest}_r\rangle$, $\langle {AttDigest}_5\rangle$, $\langle {hash}_2$, $\pi_2$, $\text{\{``Audi''\}}$, ${AttDigest}_2\rangle$, $\langle{hash}_6$, $\pi_6$, $\text{\{``Van''\}}$,\linebreak ${AttDigest}_6\rangle\}$. Here $\pi_2$ and $\pi_6$ are two disjoint proofs of the mismatching nodes $N_2$ and $N_6$ (shaded in Fig.~\ref{fig:intra_index}), respectively. Note that ${AttDigest}_r$ and ${AttDigest}_5$ will only be used to reconstruct the \textit{MerkleRoot} during result verification.
On the user side, the mismatch verification works by invoking \textsf{VerifyDisjoint($\cdot$)} using the \textit{AttDigest}, the disjoint set, and the proof $\pi$ in the VO. Further, in order to verify the result soundness and completeness, the user is required to reconstruct the \textit{MerkleRoot} and compare it with the one read from the block header. In our example, firstly, \textsf{VerifyDisjoint($\cdot$)} is invoked using $\langle\pi_2$, ${AttDigest}_2$, $\text{\{``Audi''\}}\rangle$ and $\langle\pi_6$, ${AttDigest}_6$, $\text{\{``Van''\}}\rangle$ to prove that nodes $N_2$ and $N_6$ indeed mismatch the query. After that, the user computes $\textsf{hash}(o_1)$ using the returned result, and ${hash}_5=\textsf{hash}(\textsf{hash}(o_1)~|~{hash}_2~|~ {AttDigest}_5)$, ${hash}_r=\textsf{hash}({hash}_5$ $|~{hash}_6~|~{AttDigest}_r)$ based on the VO. Finally, the user checks the newly computed ${hash}_r$ against the \textit{MerkleRoot} in the block header.

\begin{algorithm}[t]
    \caption{Query w. Intra-Index (by the SP)}\label{alg:intra-index-query}
    \small
    \SetKw{is}{is}
    \SetKwFunction{IntraIndexQuery}{IntraIndexQuery}
    \Fn{\IntraIndexQuery{$root,q$}}{%
        \KwIn{Intra-Index root $root$, Query condition $q = \langle \Upsilon \rangle$}
        \KwOut{Query Result $R$, Verification Object VO}
        Create an empty queue $queue$\;
        $queue$.enqueue($root$)\;
        \While{$queue$ \is not empty}{
        	$n\gets queue.\textsf{dequeue}()$\;
        	\eIf{$W_n \ matches \ q$ }{
        		\lIf{$n$ \is $a\ leaf\ node$}
        		{Add $o_n$ to $R$}
        		\Else{
                    Add $\langle {AttDigest}_n \rangle$ to VO\;
                    $queue$.enqueue($n$.children)\;}
        	}
        	{%
        		Find query condition set $\Upsilon_i$ w.r.t.\ $W_n$ (see Alg.~\ref{alg:single-obj-query})\;
            	$\pi \gets \textsf{ProveDisjoint}(W_n, \Upsilon_i, pk)$\;
            	add $\langle {hash}_n, \pi, \Upsilon_i, {AttDigest}_n\rangle$ to VO\;
        	}
        }
    \Return{$\langle R,\textrm{VO} \rangle$}\;
    }
\end{algorithm}

\subsection{Inter-Block Index}\label{sec:inter_index}

Besides similar objects within the same block, the objects across blocks
may also share similarity and mismatch a query due to the same reason. Based on this observation, we build an inter-block index that uses a skip list to further optimize the query performance.

\begin{figure*}[t]
    \centering
    \resizebox{0.95\linewidth}{!}{\begin{tikzpicture}
    \tikzstyle{block node}=[
    rectangle split,
    rectangle split horizontal,
    rectangle split parts=5,
    rectangle split ignore empty parts,
    rectangle split part align=base,
    draw
    ]
    \tikzstyle{list node}=[
    rectangle split,
    rectangle split,
    rectangle split parts=4,
    rectangle split ignore empty parts,
    rectangle split part align=base,
    draw
    ]

    \node[block node ] (block0) {
        \nodepart{one} PreBkHash
        \nodepart{three} MerkleRoot
        \nodepart{four} SkipListRoot
    };

    \node[list node, scale=0.95, below=0.6cm of block0.four south] (skiplist0) {
        \nodepart{one} $L_2$
        \nodepart{two} $L_4$
        \nodepart{four} $\dots$
    };

    \draw[-latex] (skiplist0) -- (block0.four south);

    \node[block node, scale=0.8, left=0.8cm of skiplist0.one west] (skiplist0-data1) {
        \nodepart{one} \textit{PreSkippedHash}$_{L_2}$
        \nodepart{two} $W_{L_2}$
        \nodepart{three} \textit{AttDigest}$_{L_2}$
    };
    \draw[dashed,-latex] (skiplist0-data1.three east) -- (skiplist0.one west);

    \node[block node, scale=0.8,  left=0.8cm of skiplist0.two west] (skiplist0-data2) {
        \nodepart{one} \textit{PreSkippedHash}$_{L_4}$
        \nodepart{two} $W_{L_4}$
        \nodepart{three} \textit{AttDigest}$_{L_4}$
    };
    \draw[dashed,-latex] (skiplist0-data2.three east) -- (skiplist0.two west);

    \node[scale=0.8,  left=0.8cm of skiplist0.two west] (skiplist0-data3) {$\dots$};

    \node[draw,fit=(block0)(skiplist0)(skiplist0-data1)(skiplist0-data3)] (box0) {};
    \node[scale=1.2,above=0.05cm of box0] (block0-label) {block$_i$};

    \foreach \i/\prev/\label in {1/0/{i-2},2/1/{i-4}}{%
        \node[left=1cm of block\prev] (omit\i) {$\dots$};

        \node[block node, left=1cm of omit\i] (block\i) {
            \nodepart{one} PreBkHash
            \nodepart{two} $\dots$
            \nodepart{three} $\dots$
        };
        \node[below=0.2cm of block\i.three south] (block\i-content) {$\dots$};
        \draw[-latex] (block\i-content) -- (block\i.three south);

        \node[draw,fit=(block\i)(block\i-content)] (box\i) {};
        \node[scale=1.2,above=0.05cm of box\i] (block\i-label) {block$_{\label}$};

        \draw[-latex] (block\prev.west) -- (omit\i.east);
        \draw[-latex] (omit\i.west) -- (box\i.east |- block\i.east);
    }

    \node[left=1cm of block2] (prev) {$\dots$};
    \node[right=1cm of block0] (next) {$\dots$};
    \draw[-latex] (block2.west) -- (prev.east);
    \draw[-latex] (next.west) -- (box0.east |- block0.east);
    \draw[-latex] (skiplist0-data1.west) -| (box1.south);
    \draw[-latex] (skiplist0-data2.west) -| (box2.south);
\end{tikzpicture}}
    \caption{Inter-Block Index}\label{fig:inter_index}
\end{figure*}
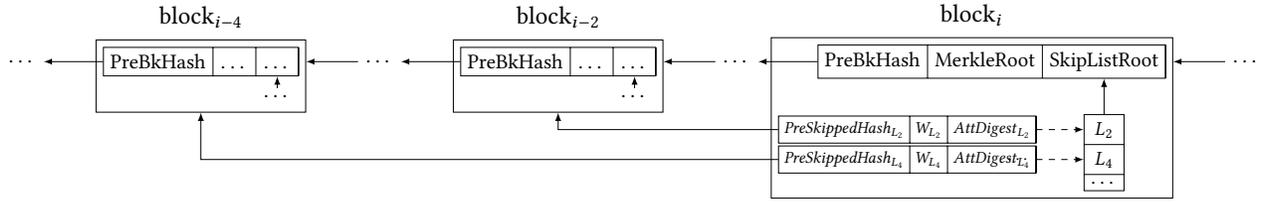

As shown in Fig.~\ref{fig:inter_index}, the inter-block index consists of multiple skips, each of which skips an exponentially number of previous blocks. For example, the list may skip previous $2, 4, 8, \cdots$ blocks. For each skip, it maintains three components: the hash of all skipped blocks (denoted by \linebreak $\textit{PreSkippedHash}_{L_k}$), the sum of the attribute multisets for the skipped blocks (denoted by $W_{L_k}$), and the corresponding accumulative value w.r.t.\ $W_{L_k}$ (denoted by $\textit{AttDigest}_{L_k}$). Note that here we use the summation of attribute multisets to enable online aggregate authentication in Section~\ref{sec:query_index}. Finally, the inter-block index is written into the block using an extra field \textit{SkipListRoot}, which is defined as:
\begin{itemize}
\item  $\textit{SkipListRoot} = \textsf{hash}({hash}_{L_2}~|~{hash}_{L_4}~|~{hash}_{L_{8}} |\cdots) $
\item ${hash}_{L_k} = \textsf{hash}(\textit{PreSkippedHash}_{L_k}~|~\textit{AttDigest}_{L_k}) $
 \item  $\textit{AttDigest}_{L_k} = acc(W_{L_{k}})$
 \item  $ W_{L_{k}} =\textstyle\sum_{j=i-k+1}^{i} W_{j} $
\end{itemize}

During the query processing, an eligible skip may be used to represent multiple blocks which do not contribute to query results due to the same reason of mismatching. As the user can avoid accessing these skipped blocks, the verification cost can be reduced.

\begin{algorithm}[t]
    \caption{Query w. Inter-Index (by the SP)}\label{alg:inter-index-query}
    \small
    \SetKw{to}{to}
    \SetKw{go}{go}
    \SetKw{from}{from}
    \SetKw{is}{is}
    \SetKw{not}{not}
    \SetKw{and}{and}
    \SetKwFunction{InterIndexQuery}{InterIndexQuery}
    \Fn{\InterIndexQuery{$block,q$}}{%
        \KwIn{Current Block $block$, Query condition $q = \langle \Upsilon \rangle$}
        $FindJump\gets False$\;
        \For{$L_{i}$ \from $L_{max}$ \to $L_{min}$ of $block$.SkipList}{
            \If{$W_{L_i}$ does not match $q$ \and $FindJump\neq True$}{
                Find query condition set $\Upsilon_i$ w.r.t.\ $W_n$\;
                $\pi \gets \textsf{ProveDisjoin}(W_{L_i}, \Upsilon_i, pk)$\;
                add $\langle \textit{PreSkippedHash}_{L_i},\pi,\Upsilon_i, \textit{AttDigest}_{L_{i}}\rangle$ and all $hash_{L_{j}}\in SkipList,j\neq i$ to VO\;
                $block\gets Skip(i)$; $FindJump\gets True$\;
            }
        }
        \If{$FindJump=False$}
        {\IntraIndexQuery{$block.root,q$}\;
            $block\gets block.prev$;
        }
        \InterIndexQuery{$block,q$}\;

    }
\end{algorithm}

Algorithm~\ref{alg:inter-index-query} shows the query processing procedure with the inter-block index. We start with the latest block in the query time window.
We iterate the skip list from the maximum skip to the minimum skip.
If the multiset of a skip $W_{L_{i}}$ does not match the query condition, it means that all the skipped blocks between the current block and the previous $i$-th block do not contain matching results. Therefore,  \textsf{ProveDisjoint}($\cdot$) is invoked and output the mismatch proof $\pi_i$. Then, $\langle \textit{PreSkippedHash}_{L_{i}}$, $\pi_i$, $\Upsilon_{i}$, $\textit{AttDigest}_{L_{i}} \rangle$ are added to the VO. The user can use this proof to verify that the skipped blocks indeed mismatch the query. Meanwhile, other hashes except $hash_{L_{i}}$ are also added to the VO.
If we fail to find mismatching blocks during the iteration, the function \IntraIndexQuery{$\cdot$} (Algorithm~\ref{alg:intra-index-query}) is invoked for the current block and then the previous block is examined next. If we successfully find a mismatch skip, the corresponding preceding block will be examined next. The function \InterIndexQuery{$\cdot$} is invoked recursively until we complete checking all the blocks within the query window. Note that we can combine the intra-block index and inter-block index to maximize performance since they are not in conflict.

\subsection{Online Batch Verification} \label{sec:online_batch}

Recall that the proposed intra-block index attempts to cluster the objects of the same block in a way to maximize the proofing efficiency of mismatching objects. Nevertheless, some objects/nodes indexed in different blocks or even different subtrees of the same block may also share the same reason of mismatching. Therefore, it would be beneficial to aggregate such objects/nodes online for more efficient proofing. To do so, the \textsf{Sum}($\cdot$) primitive introduced by Construction 2 in Section~\ref{sec:acc}, which outputs the accumulative value of the aggregated multiset when given multiple accumulative values, can be applied.
In our running example shown in Fig.~\ref{fig:intra_index}, suppose that $o_2$ and $o_4$ share the same reason for mismatching a query condition $(\text{``Benz''})$. Then, the SP can return $\pi = \textsf{ProveDisjoint}(W_2 + W_4, \text{\{``Benz''\}}, pk)$ and $\textit{AttDigest}_{2,4} = \textsf{Sum}(\textsf{acc}(W_2), \textsf{acc}(W_4))$. And the user can apply $\textsf{VerifyDisjoint}(\textit{AttDigest}_{2,4}, \textsf{acc}(\text{\{``Benz''\}}),$ $\pi, pk)$ to prove that these two objects mismatch in a batch.

\section{Verifiable Subscription Queries}\label{sec:query_index}

A subscription query is registered by the query user and continuously processed until it is deregistered. Upon seeing a newly confirmed block, the SP will need to publish the results to registered users, together with VOs. In this section, we first propose a query index to efficiently handle a large number of subscription queries (Section~\ref{sec:q_index}). After that, we develop a lazy authentication optimization that delays mismatch proofs to reduce the query verification costs (Section~\ref{sec:lazy_opt}).

\subsection{Query Index for Scalable Processing} \label{sec:q_index}

As discussed earlier, the majority of the query processing overhead comes from generating the proofs for mismatching objects at the SP\@. Fortunately, a mismatching object could have the same reason of mismatching for different subscription queries. Thus, a mismatch proof can be shared by such queries. Inspired by~\cite{Chen:SIGMOD13}, we propose to build an inverted prefix tree, called \emph{IP-Tree}, over subscription queries. It is essentially a prefix tree with reference to inverted files for both the numerical range condition and also the Boolean set~condition.

\textbf{Prefix Tree Component}. To index the numerical ranges of all subscription queries, the IP-Tree is built on the basis of a grid tree such that each tree node is represented by a CNF Boolean function (see Section~\ref{sec:extendrange}). For example, the grid node $N_1$ in Fig.~\ref{fig:invert_index}, corresponding to the upper-left cell $([0, 2], [1, 3])$, is denoted by $\{0$$*_1 \land 1*_2\}$. The root node of the prefix tree covers the entire range space of all subscription~queries.

\textbf{Inverted File Component}. Each node of the IP-Tree is associated with an inverted file that is constructed based on the subscription queries indexed under the node. There are two subcomponents for each inverted file:
\begin{itemize}
\item \textbf{Range Condition Inverted File (RCIF)}. Each entry in the $RCIF$ has two attributes: query $q_i$ and its cover type (i.e., full or partial). All the queries in the $RCIF$ intersect the numerical space $\mathcal{S}$ of the node. The cover type indicates whether  $q_i$ fully covers or partially covers $\mathcal{S}$. The $RCIF$ is used to check the mismatch of the numerical range condition.
\item \textbf{Boolean Condition Inverted File (BCIF)}. The $BCIF$  records only the queries that fully cover the node's space. Each entry in the $BCIF$ consists of two attributes: query condition set $\Upsilon$ and corresponding queries. The $BCIF$ is used to check the mismatch of the Boolean set~condition.
\end{itemize}
\begin{figure}[t]
    \centering
    \resizebox{\linewidth}{!}{\begin{tikzpicture}
    \node[matrix](grid){
        \draw [step=2cm] (0,0) grid (4,4);
        \draw (0,3) -- (2, 3);
        \draw (1, 4) -- (1, 2);
        \draw[line width=1.5pt,dashed,red]
        (0.05,2.05) rectangle (1.95,3.95) node[anchor=north east] {$q_1$};
        \draw[purple]
        (-0.05,-0.05) rectangle (2.05,4.05) node[anchor=south east] {$q_2$};
        \draw[Green,line width=1.2pt]
        (0.1,2.1) rectangle (0.95,2.95) node[anchor=south east] {$q_3$};
        \draw[line width=1.2pt, blue]
        (2.05,0.05) rectangle (3.95,3.95) node[anchor=north east] {$q_4$};

        \node[purple,circle,fill,inner sep=1pt] (object) at (0.5, 2.5) {};
        \node [left=0] at (1.9,2.2) {$N_{1}$};
        \node [left=0] at (3.9,2.2) {$N_{2}$};
        \node [left=0] at (1.9,0.3) {$N_{3}$};
        \node [left=0] at (3.9,0.3) {$N_{4}$};
        \draw [-latex] (-0.5,-0.2) -- (4.5,-0.2) node[right]{x};
        \draw [-latex] (-0.2,-0.5) -- (-0.2,4.5) node[above]{y};
        \foreach \x in {0,1,2,3,4} {%
                \draw ($(\x,-0.2) + (0,-2pt)$) -- ($(\x,-0.2) + (0,2pt)$);
            }

        \foreach \y in {0,1,2,3,4} {%
                \draw ($(-0.2,\y) + (-2pt,0)$) -- ($(-0.2,\y) + (2pt,0)$);
            }
        \node[] at(0.5,-0.5){$0$};
        \node[] at(1.5,-0.5){$1$};
        \node[] at(2.5,-0.5){$2$};
        \node[] at(3.5,-0.5){$3$};
        \node[] at(-0.5,0.5){$0$};
        \node[] at(-0.5,1.5){$1$};
        \node[] at(-0.5,2.5){$2$};
        \node[] at(-0.5,3.5){$3$};
        \\
    };

    \node[matrix,right=0.1cm of grid](tree){
        \path[every node/.style={draw,circle}]
        node (n0) {$N_0$}
        child { node (n1) {$N_1$}
                child { node (n5) {$N_5$} }
                child { node (n6) {$N_6$} }
                child { node (n7) {$N_7$} }
                child { node (n8) {$N_8$} }
            }
        child {node (n2) {$N_2$}}
        child { node (n3) {$N_3$} }
        child { node (n4) {$N_4$} };
        \\
    };

    \node[below=0.2cm of grid] (oi) {
        $\begin{aligned}
                o_i & = \langle t_i, (0, 2), \{ \text{``Van''}, \text{``Benz''} \} \rangle         \\
                    & = \langle t_i, \{ {00}_1, {10}_2, \text{``Van''}, \text{``Benz''} \} \rangle
            \end{aligned}$
    };

    \node[left=-0.45cm of oi.north west] (oi_anchor1){};
    \node[below=0.01cm of oi_anchor1] (oi_anchor2){};
    \draw[purple,-latex,dashed] (oi_anchor2.north west) -- (object);

    \node[below=0.2cm of oi] (db) {
        \begin{tabu}{|c|c|l|}
            \hline
            \rowfont{\bfseries}
            Query & Range           &
            Boolean Condition                                                          \\
            \hline
            $q_1$ & $[(0,2),(1,3)]$ & $\{\text{``Van''} \land \text{``Benz''}\}$   \\
            \hline
            $q_2$ & $[(0,0),(1,3)]$ & $\{\text{``Van''} \land \text{``BMW''}\}$    \\
            \hline
            $q_3$ & $[(0,2),(0,2)]$ & $\{\text{``Sedan''} \land \text{``Audi''}\}$ \\
            \hline
            $q_4$ & $[(2,0),(3,3)]$ & $\{\text{``Sedan''} \land \text{``Benz''}\}$ \\
            \hline
        \end{tabu}
    };

    \node [anchor=north, below=0.6cm] at (tree.south) (text){$\textsf{Grid Cell}: [(0,2),(1,3)]\rightarrow \{{0*}_1 \land {1*}_2\}$};

    \node[anchor=north west] at (text.south west) (text1) {\textsf{RCIF:}};

    \node[anchor=north west] at (text1.north east) (t1) {
        \begin{tabular}{|c|c|}
            \hline
            \textbf{Query} & \textbf{Cover Type} \\
            \hline
            $q_1$          & full                \\
            $q_2$          & full                \\
            $q_3$          & partial             \\

            \hline
        \end{tabular}
    };

    \node[anchor=north west] at (t1.south west) (t2) {
        \begin{tabular}{|c|l|}
            \hline
            \textbf{Query Condition Set $\Upsilon$ } & \textbf{Queries} \\
            \hline
            $\{\text{``Van''}\}$                     & $q_1,q_2$        \\
            $\{\text{``Benz''}\}$                    & $q_1$            \\
            $\{\text{``BMW''}\}$                     & $q_2$            \\

            \hline
        \end{tabular}
    };

    \node[anchor=north east] at (t2.north west) (text2) {\textsf{BCIF:}};

    \node[draw=black, thick, rounded corners, fit=(text)(t1)(t2)] (tables) {};

    \node[draw,black,dashed,circle,rounded corners,fit=(n1)] (n1-box) {};
    \draw[-latex,black,dashed] (n1-box.west) to [out=180,in=45] ($(n5) + (-0.2cm, 0.9cm)$) to [out=-135,in=150] (tables.north west);
\end{tikzpicture}}
    \caption{Inverted Prefix Tree}\label{fig:invert_index}
\end{figure}
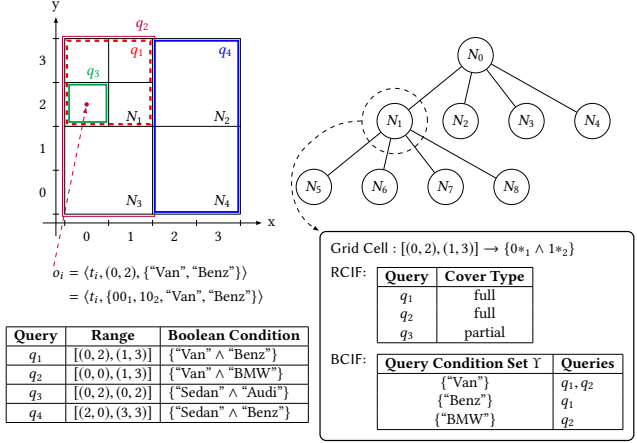
We use Fig.~\ref{fig:invert_index} as an example to illustrate how to construct the IP-Tree. It is built in a top-down fashion by the SP. We first create the root node and add all queries to its $RCIF$ as partial-cover queries. We then split the root node and create four equally-spaced child nodes. For each child node, if a query fully or partially covers the node's space, it will be added to the node's $RCIF$. Also, the equivalence sets of a full-cover query will be added to the node's $BCIF$.
Take $N_1$ as an example. While queries $q_{1}$ and $q_{2}$ fully cover this node, query $q_3$ only partially covers it. Thus, the $RCIF$ contains three intersection queries $q_1$, $q_2$, and $q_3$. The cover types of $q_1$ and $q_2$ are \textit{full} and for $q_3$ the type is \textit{partial}. As for the $BCIF$ of $N_1$, $q_1$ and $q_2$ share the equivalence set $\{\text{``Van''}\}$, and the sets $\{\text{``Benz''}\}$ and $\{\text{``BMW''}\}$ correspond to queries $q_1$ and $q_2$, respectively.
Next, since $q_3$ only partially covers $N_1$, we further split $N_1$ into four sub-cells. As $q_3$ fully covers $N_7$, it is added to the $RCIF$ and $BCIF$ of $N_7$. The algorithm terminates when no partial query is found in any leaf node. When a query is registered or deregistered, we update the IP-Tree's nodes corresponding to the numerical range of the query. We may also split or merge the tree nodes if necessary. Note that, to prevent the tree from becoming too deep, we switch back to the case without the IP-Tree when the tree depth reaches some pre-defined threshold.

With the IP-Tree index, the subscription queries can be processed as a tree traversal. We first use an example of single object to illustrate the basic idea. Upon arrival of a new object $o$, the IP-Tree is traversed along the path from the root to the leaf node that covers $o$. For any node $n_q$ on the traversal path, the associated queries can be found from $n_q$'s $RCIF$. These queries can be classified into three categories: (1) a full-cover query whose equivalence set(s) in $n_q$'s $BCIF$ match $o$ (thus $o$ is added as a result of this query); (2) a full-cover query whose equivalence set(s) in $n_q$'s $BCIF$ mismatch $o$ (thus \textsf{ProveDisjoint}($\cdot$) is invoked and a disjoint proof is generated for this query); (3) a partial-cover query (no further action is needed). In addition, we identify the queries that appear in $n_q$'s parent's $RCIF$ but not in $n_q$'s. Those queries mismatch the numerical range condition for $o$ and thus also a disjoint proof is generated for them. Next, $n_q$'s child node will be processed and this process continues until we reach a leaf node or all queries have been classified as matching or mismatching. Consider a new object $o_i = \langle t_i, (0,2),   \{\text{``Van''},\text{``Benz''}\} \rangle = \langle t_i, \{00_1, 10_2, \text{``Van''},\text{``Benz''}\} \rangle$ shown in Fig.~\ref{fig:invert_index}. At $N_1$, $q_1$ is classified as a matching query, $q_2$ and $q_4$ mismatch because of the Boolean set condition and the numerical range condition, respectively, whereas $q_3$ is not confirmed as mismatching until we check $N_1$'s child node $N_7$.

This idea can be easily extended to a new block of objects that are indexed by an intra-block index. We start from the root of the intra-block index. For any index node $n_b$, we treat it as a \textit{super object} and apply the above query processing procedure. The only difference is that if a full-cover query is classified as matching, we cannot immediately return the current node $n_b$ as a query result, but to further recursively check its child nodes until reaching the leaf nodes. In the interest of space, the pseudo codes of the detailed algorithms are given in Appendix~\ref{sec:ip_algorithms}.

\begin{algorithm}[t]
    \caption{\small Subscription Query w. Lazy Authentication (by the SP)}\label{alg:subscription-interindex}
    \small
    \SetKw{to}{to}
    \SetKw{go}{go}
    \SetKw{from}{from}
    \SetKw{is}{is}
    \SetKw{not}{not}
    \SetKw{and}{and}
    \SetKwFunction{Subscribe}{SubscribeInterIndexQuery}
    \Fn{\Subscribe{$block,q,\textrm{VO},s$}}{%
        \KwIn{Block $block$, Query condition $q = \langle \Upsilon \rangle$, Partial VO, Stack $s$ }
        $W_r\gets block.root.W_{r}$\;
        \eIf{$W_r$ matches $q$}{
            $\langle R, \textrm{VO}_b \rangle\gets$\IntraIndexQuery{$block.root,q$}\;
            $\textrm{VO}\gets \textrm{VO}+\textrm{VO}_b$;
            Send $\langle R, \textrm{VO} \rangle$ to user\;
            Empty partial $\textrm{VO}$ and $s$\;
        }{
            $\textrm{VO}_b\gets$\IntraIndexQuery{$block.root,q$}\;
            \eIf{$W_r$ has the same mismatch attributes with $s$}{
                Find the maxiumn skip $L_{i}$ s.t. it covers $m$ elements on top of $s$\;%
                \uIf{$L_{i}$ is found}{
                $\langle {block}_i, \textrm{JumpDistance}_i\rangle\gets{}s[i],\forall i \in \{m\}$\;
                Pop $m$ elements of $s$\;
                Rewind partial $\textrm{VO}$ to ${block}_m$\;
                Add $\textrm{VO}_b$, $\textit{AttDigest}_{L_{i}}$ and other hashes to VO\;
                $s.push(\langle block, \sum_i^m \textrm{JumpDistance}_i \rangle)$\;
                }
                \lElse{$s.push(\langle block, 1 \rangle)$}
            }{
                Empty $s$;
                $s.push(\langle block, 1 \rangle)$\;
            }
            }
        }
\end{algorithm}

\subsection{Lazy Authentication} \label{sec:lazy_opt}
Observing that in the previous section, the results and proofs are immediately published to registered users while a new block is confirmed. In particular, even if there is no matching result for a query, the mismatch proofs are still computed and sent. This approach is good for real-time applications. For applications that do not have such real-time requirements, we propose a \emph{lazy authentication} optimization, in which the SP returns the result only when there is a matching object (or the time since the last result has passed a threshold).

In this approach, the VO should prove that the current object is a match and all other objects since the last result mismatch the query. To achieve this, we may simply wait for the matching result and invoke a time-window query to compute the mismatch proofs on the fly. However, this method can only generate the mismatch proofs for each query separately and is incapable to take advantage of the proofs shared by different subscription queries. Moreover, this method leaves the burden of proofing all to the time when there is a matching result. To address these issues, we propose a new method that makes use of the inter-block index to incrementally generating mismatch proofs.

Using the inter-block index to answer a subscription query is completely different from doing that to a time-window query. The reason is that we can back traverse the blockchain and use the skip list to aggregate proofs in a time-window query. However, we cannot do so for a subscription query because new blocks are not yet available and we do not know whether or not future objects will share the same mismatch conditions. As such, we introduce a \textit{stack} to facilitate tracking the arrived blocks that share the same mismatch conditions. The basic idea is to use the skip list to find the maximum skip distance $L_i$ such that it covers $m$ elements on top of the stack. The \textit{AttDigest}s of these blocks are replaced with $\textit{AttDigest}_{L_{i}}$. Thanks to Construction 2 in Section~\ref{sec:acc}, the disjoint proofs can be aggregated online by invoking \textsf{ProofSum}($\cdot$). For example, we have two mismatching $block_i$ and $block_{i-1}$ in the stack and there is a skip with distance 2. Then, the SP can replace their proofs by an aggregate proof computed from \textsf{ProofSum}$(\pi_i, \pi_{i-1})$. In this way, the SP does not need to compute the set disjoint proofs from scratch when a matching result is found.
The detailed procedure is described in Algorithm~\ref{alg:subscription-interindex}.

\section{Security Analysis}\label{sec:security_analysis}

This section performs a security analysis on the multiset accumulators and query authentication algorithms.

\subsection{Analysis on Multiset Accumulators}\label{sec:security_analysis_acc}

We first present a formal definition of the security notion for multiset accumulators and set disjoint proofs.

\begin{definition}[Unforgeability~\cite{papamanthou2011optimal}]\label{def:acc-unf}
    We say a multiset accumulator is unforgeable if the success probability of any polynomial-time adversary is negligible in the following experiment:
    \begin{itemize}
        \item Run $(sk, pk) \gets \textsf{KeyGen}(1^\lambda)$ and give the public key $pk$ to the adversary;
        \item The adversary outputs two multisets $X_1$ and $X_2$, along with a set disjoint proof $\pi$.
    \end{itemize}
    We say the adversary succeeds if $\textsf{VerifyDisjoint}(\textsf{acc}(X_1)$,\linebreak $\textsf{acc}(X_2)$, $\pi$, $pk)$ outputs 1 and $X_1 \cap X_2 \neq \emptyset$.
\end{definition}

This property ensures that the chance for a malicious SP to forge a set disjoint proof is negligible, which serves as a foundation for the security of our proposed query authentication algorithms. We now show that our constructions of the accumulator indeed satisfy the desired security requirement.

\begin{restatable}{theorem}{accsecuritytheorem}\label{THE:ACC-SEC}
    The constructions of the multiset accumulator presented in Section~\ref{sec:acc} satisfy the security property of the unforgeability as defined in Definition~\ref{def:acc-unf}.
\end{restatable}
\vspace{-3ex}
\begin{proof}
    See Appendix~\ref{prof:acc-security} for a detailed proof.
\end{proof}

\vspace{-2ex}
\subsection{Analysis on Query Authentication}\label{sec:security_analysis_alg}

The formal definition of the unforgeability for our query authentication algorithms is given below:

\begin{definition}[Unforgeability]\label{def:alg-unf}
    We say our proposed query authentication algorithms are unforgeable if the success probability of any polynomial-time adversary is negligible in the following experiment:
    \begin{itemize}
        \item Run the ADS generation and give all objects $\{o_i\}$ to the adversary;
        \item The adversary outputs a query $q$ (either time-window or subscription query), a result $R$, and a VO\@;
    \end{itemize}
    We say the adversary succeeds if the VO passes the result verification and one of the following results is true:
    \begin{itemize}
        \item $R$ contains an object $o^*$ such that $o^* \notin \{o_i\}$;
        \item $R$ contains an object $o^*$ such that $o^*$ does not satisfy the query $q$;
        \item There exists an object $o_x$ in the query time window or subscription period, which is not in $R$ but satisfies $q$.
    \end{itemize}
\end{definition}

This property ensures that the chance for a malicious SP to forge an incorrect or incomplete result is negligible. We can show that our proposed query authentication algorithms indeed satisfy the desired security requirement.

\begin{restatable}{theorem}{algsecuritytheorem}\label{THE:ALG-SEC}
    Our proposed query authentication algorithms satisfy the security property of the unforgeability as defined in Definition~\ref{def:alg-unf}.
\end{restatable}
\vspace{-3ex}
\begin{proof}
    See Appendix~\ref{prof:alg-security} for a detailed proof.
\end{proof}

\section{Performance Evaluation}\label{sec:experiment}

In this section, we evaluate the performance of the vChain framework for time-window queries and subscription queries. Three datasets are used in the experiments:
\begin{itemize}
\item  \textbf{Foursquare} (\textsf{4SQ}) \cite{yang2015participatory}: The 4SQ dataset contains 1M data records, which are the user check-in information. We pack the records within a 30s interval as a block and each object has the form of $\langle timestamp$, $[longitude$, $latitude]$, $\{\textit{check-in place's keywords}\} \rangle$. On average, each record has 2 keywords.

\item \textbf{Weather} (\textsf{WX}): The WX dataset contains 1.5M hourly weather records for 36 cities in US, Canada, and Israeli  during 2012-2017.\footnote{\url{https://www.kaggle.com/selfishgene/historical-hourly-weather-data/}} For each record, it contains seven numerical attributes (such as \textit{humidity} and \textit{temperature}) and one \textit{weather description} attribute with 2 keywords on average. The records within the same hour interval are packed as a block.

\item \textbf{Ethereum} (\textsf{ETH}): The ETH transaction dataset is extracted from the Ethereum blockchain during the period from Jan 15, 2017 to Jan 30, 2017.\footnote{\url{https://www.ethereum.org/}} It contains 90,000 blocks with 1.12M transaction records.
Each transaction is in the form of $\langle timestamp$, $amount$, $\{addresses\} \rangle$, where $amount$ is the amount of Ether transferred and $\{addresses\}$ are the addresses of senders and receivers.
Most transactions have two addresses.
\end{itemize}
Note that the time intervals of the blocks in \textsf{4SQ}, \textsf{WX}, and \textsf{ETH} are roughly 30s, 1 hour, and 15s, respectively.

The query user is set up on a commodity laptop computer with Intel Core i5 CPU and 8GB RAM, running on CentOS~7 with a single thread. The SP and the miner are set up on a x64 blade server with dual Intel Xeon 2.67GHz, X5650 CPU and 32 GB RAM, running on CentOS 7. The experiments are written in C++ and the following libraries are used: MCL for bilinear pairing computation,\footnote{MCL: \url{https://github.com/herumi/mcl/}} Flint for modular arithmetic operations, Crypto++ for 160-bit SHA-1 hash operations, and OpenMP for parallel computation. Also, the SP runs with 24 hyperthreads to accelerate the query processing.

To evaluate the performance of verifiable queries in vChain, we mainly use three metrics:
\begin{enumerate*}[label=(\roman*)]
    \item query processing cost in terms of SP CPU time,
    \item result verification cost in terms of user CPU time, and
    \item size of the VO transmitted from the SP to the user.
\end{enumerate*}
For each experiment, we randomly generate 20 queries and report the average results. By default, we set the selectivity of the numerical range to 10\% (for \textsf{4SQ} and \textsf{WX}) and 50\% (for \textsf{ETH}) and employ a disjunctive Boolean function with a size of 3 (for \textsf{4SQ} and \textsf{WX}) and 9 (for \textsf{ETH}). For \textsf{WX}, two attributes are involved in each range~predicate.

\vspace{-2ex}
\subsection{Setup Cost}

\FloatBarrier

\begin{table}[t]
    \centering
    \caption{Miner's Setup Cost}\label{tab:setup}
    \begin{adjustbox}{max width=.78\linewidth}
    \begin{threeparttable}[b]
    \small
    \begin{tabular}{cccccccc}
        \toprule
        \multirow{2}{*}{\textbf{Dataset}} & \multirow{2}{*}{\textbf{Acc}} &
        \multicolumn{2}{c}{nil} &
        \multicolumn{2}{c}{intra} &
        \multicolumn{2}{c}{both} \\
        \cmidrule{3-8}
                                   & & T & S & T & S & T & S \\
        \midrule
        \textsf{4SQ} & acc1 & 0.17 & 2.12 & 0.65 & 10.7 & 12.5 & 11.1 \\
        \textsf{4SQ} & acc2 & 0.06 & 2.12 & 0.26 & 10.7 & 1.16 & 11.1 \\
        \textsf{WX} & acc1 & 0.16 & 1.55 & 0.52 & 7.38 & 1.01 & 7.68 \\
        \textsf{WX} & acc2 & 0.05 & 1.55 & 0.16 & 7.38 & 0.20 & 7.68 \\
        \textsf{ETH} & acc1 & 0.01 & 0.55 & 0.07 & 2.60 & 0.87 & 2.93 \\
        \textsf{ETH} & acc2 & 0.14 & 0.55 & 0.30 & 2.60 & 0.13 & 2.93 \\
        \bottomrule
    \end{tabular}
    \begin{tablenotes}
        \item[] {\footnotesize T: ADS construction time (s/block);~~~S: ADS size (KB/block)}
     \end{tablenotes}
    \end{threeparttable}
    \end{adjustbox}
\end{table}

Table~\ref{tab:setup} reports the miner's setup cost, including the ADS construction time and the ADS size. Three methods are compared in our experiments:
\begin{enumerate*}[label=(\roman*)]
    \item \emph{nil}: no index is used;
    \item \emph{intra}: only intra-block index is used;
    \item \emph{both}: both intra- and inter-block indexes are used, in which the size of \emph{SkipList} in the inter-block index is set to 5.
\end{enumerate*}
Each method is implemented with two different accumulator constructions (labelled with \emph{acc1} and \emph{acc2}) presented in Section~\ref{sec:acc}. Thus, a total of six schemes are evaluated in each experiment. As expected, the ADS construction time of \emph{both} is generally longer than those of \emph{nil} and \emph{intra}, but still within 2s for most cases. Moreover, compared with \emph{acc1}, \emph{acc2} significantly reduces the construction time of \emph{both} because it supports online aggregation and hence can reuse the index of the previous block in constructing the inter-block index. Regarding the ADS size, it is independent of the accumulator used and ranges from 2.6KB to 11.1KB per block for different indexes and datasets.

We also measure the space required by the user for running a light node to maintain the block headers. For both \emph{nil} and \emph{intra}, the size of each block header is 800 bits, regardless of the dataset or accumulator. Due to the inter-block index, the block header size of \emph{both} is slightly increased to 960 bits.

\vspace{-1ex}
\subsection{Time-Window Query Performance}
To evaluate the performance for time-window queries, we vary the query window  from 2 to 10 hours for \textsf{4SQ} and \textsf{ETH} and from 20 to 100 hours for \textsf{WX}.
The results for the three datasets are shown in Figs.~\ref{exp-fig:timewindow4sq}--\ref{exp-fig:timewindoweth}, respectively. We make several interesting observations. First, as expected, the indexes substantially improve the performance in almost all metrics. In particular, for the \textsf{4SQ} and \textsf{ETH} datasets, the performance of using the indexes is at least 2X better than that with the same accumulator but without using any index. This is because the objects in these two datasets share less similarity and hence benefit more from using the indexes for pruning.
Second, the costs of the index-based schemes increase only sublinearly with enlarging the query window. This is particularly true in terms of the user CPU time for the index-based schemes using $acc2$, which supports batch verification of mismatches (see Section~\ref{sec:online_batch}). Third, comparing $intra$ and $both$, $both$ is always no worse than $intra$ except concerning the SP CPU time for the \textsf{4SQ} dataset.  On the one hand, this indicates the effectiveness of using the inter-block index. On the other hand, the reason of $both$ being worse than $intra$ in SP CPU time is mainly because in an inter-block index-based scheme, larger multisets are used as the input of a set disjoint proof, which increases the SP CPU time.
More insight on this is provided in Appendix~\ref{sec:skiplist_exp}, where we examine the impact of \emph{SkipList} size. The biggest improvement of $both$ over $intra$ is observed for the \textsf{ETH} dataset. The reason is as follows. Compared with \textsf{4SQ}, the similarity shared among the objects in \textsf{ETH} is lower; compared with \textsf{WX}, \textsf{ETH} has less objects contained in each block. For both cases, more performance improvement is gained from using the skip list in the inter-block index.

\begin{figure}[t]
    \centering
    \begin{subfigure}{.33\linewidth}
        \includegraphics[width=\linewidth]{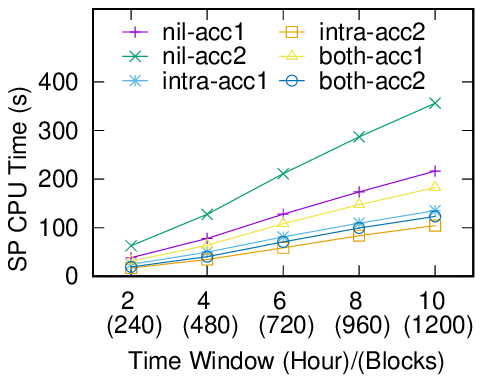}
    \end{subfigure}~%
    \begin{subfigure}{.33\linewidth}
        \includegraphics[width=\linewidth]{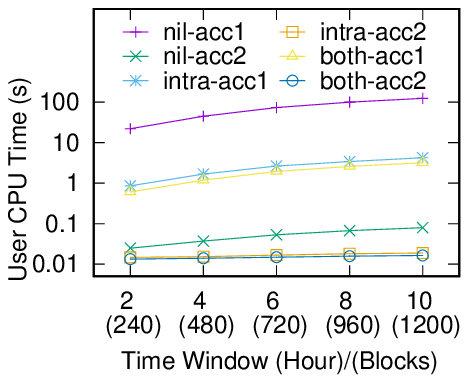}
    \end{subfigure}~%
    \begin{subfigure}{.33\linewidth}
        \includegraphics[width=\linewidth]{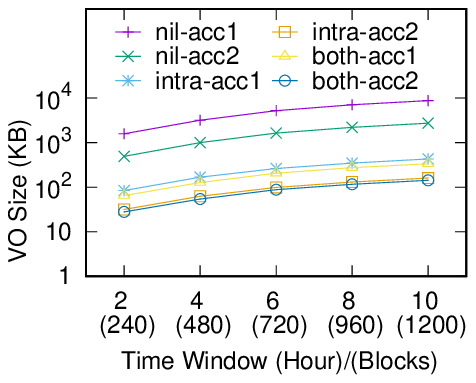}
    \end{subfigure}
    \caption{Time-Window Query Performance (\textsf{4SQ})}\label{exp-fig:timewindow4sq}
\end{figure}

\begin{figure}[t]
    \centering
    \begin{subfigure}{.33\linewidth}
        \includegraphics[width=\linewidth]{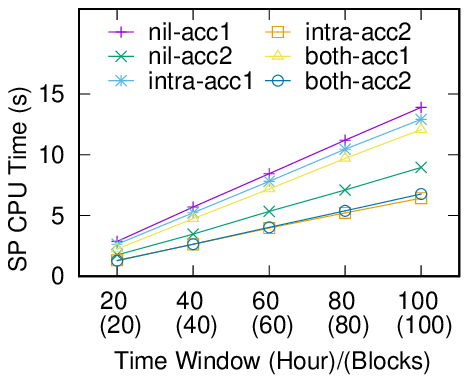}
    \end{subfigure}~%
    \begin{subfigure}{.33\linewidth}
        \includegraphics[width=\linewidth]{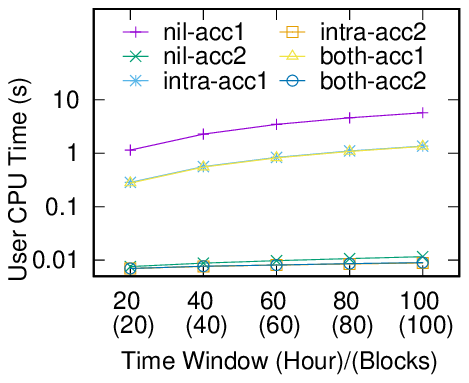}
    \end{subfigure}~%
    \begin{subfigure}{.33\linewidth}
        \includegraphics[width=\linewidth]{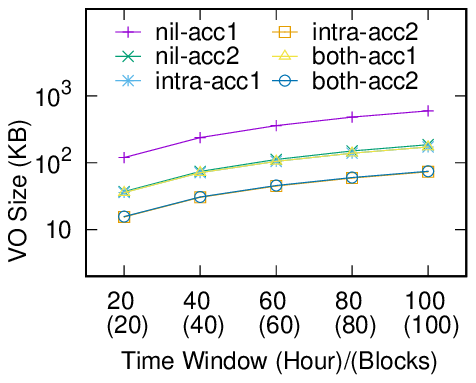}
    \end{subfigure}
    \caption{Time-Window Query Performance (\textsf{WX})}\label{exp-fig:timewindowwx}
\end{figure}

\begin{figure}[t]
    \centering
    \begin{subfigure}{.33\linewidth}
        \includegraphics[width=\linewidth]{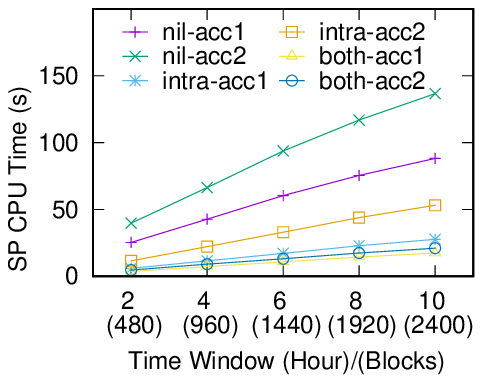}
    \end{subfigure}~%
    \begin{subfigure}{.33\linewidth}
        \includegraphics[width=\linewidth]{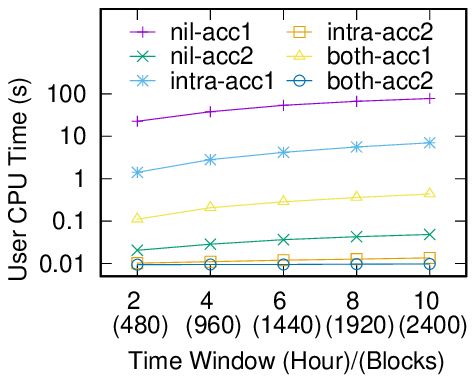}
    \end{subfigure}~%
    \begin{subfigure}{.33\linewidth}
        \includegraphics[width=\linewidth]{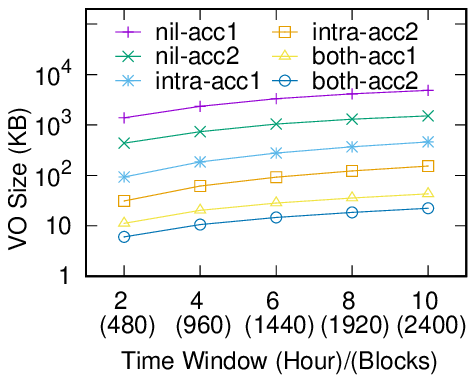}
    \end{subfigure}
    \caption{Time-Window Query Performance (\textsf{ETH})}\label{exp-fig:timewindoweth}
\end{figure}

\begin{figure}[t]
    \centering
    \begin{subfigure}{.33\linewidth}
        \includegraphics[width=\linewidth]{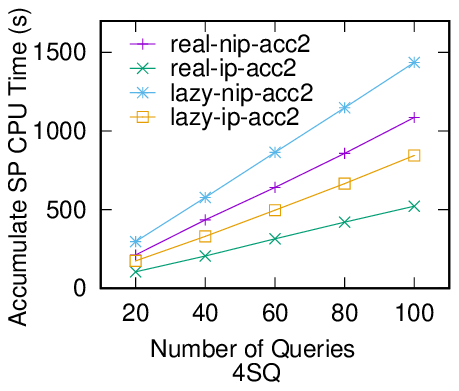}
    \end{subfigure}~%
    \begin{subfigure}{.33\linewidth}
        \includegraphics[width=\linewidth]{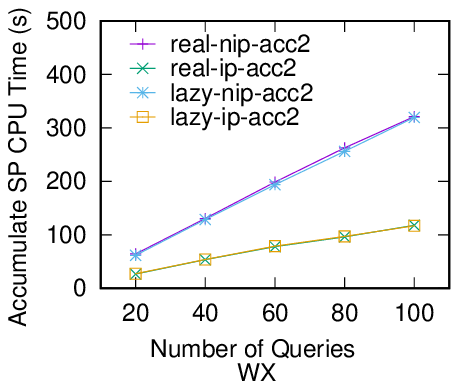}
    \end{subfigure}~%
    \begin{subfigure}{.33\linewidth}
        \includegraphics[width=\linewidth]{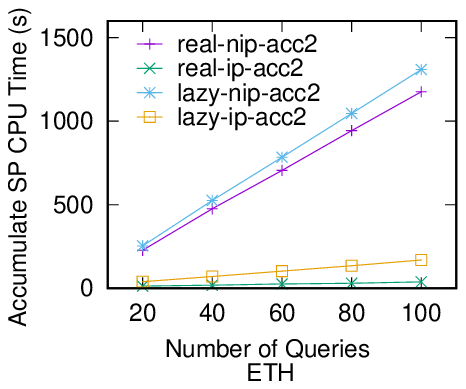}
    \end{subfigure}
    \caption{Subscription Query with IP-Tree Index}\label{exp-fig:iptree}
\end{figure}

\begin{figure}[t]
    \centering
    \begin{subfigure}{.33\linewidth}
        \includegraphics[width=\linewidth]{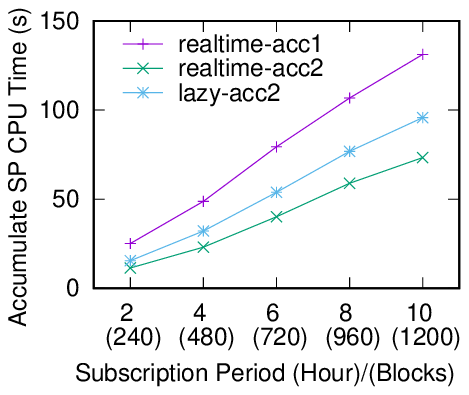}
    \end{subfigure}~%
    \begin{subfigure}{.33\linewidth}
        \includegraphics[width=\linewidth]{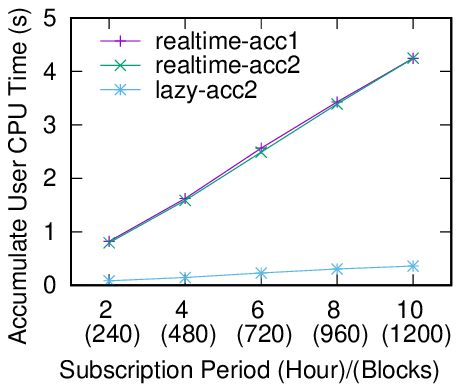}
    \end{subfigure}~%
    \begin{subfigure}{.33\linewidth}
        \includegraphics[width=\linewidth]{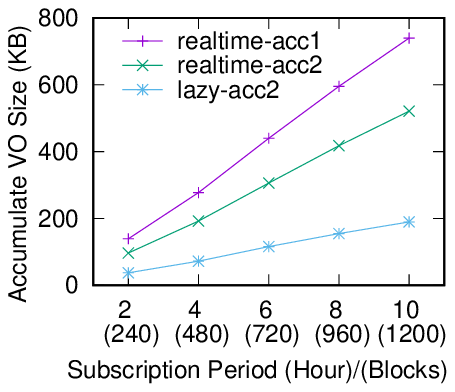}
    \end{subfigure}
    \caption{Subscription Query Performance (\textsf{4SQ})}\label{exp-fig:sub4sq}
\end{figure}

\begin{figure}[t]
    \centering
    \begin{subfigure}{.33\linewidth}
        \includegraphics[width=\linewidth]{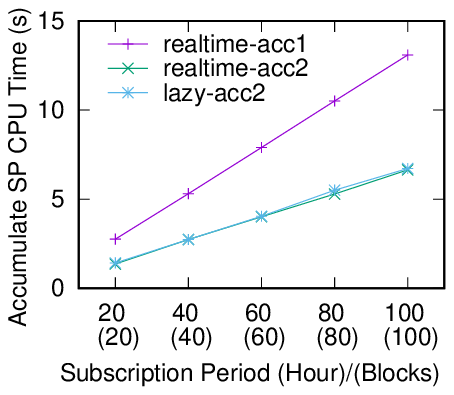}
    \end{subfigure}~%
    \begin{subfigure}{.33\linewidth}
        \includegraphics[width=\linewidth]{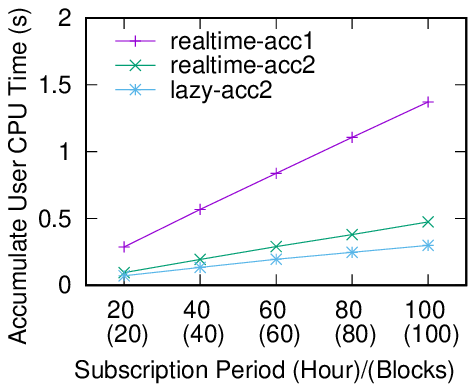}
    \end{subfigure}~%
    \begin{subfigure}{.33\linewidth}
        \includegraphics[width=\linewidth]{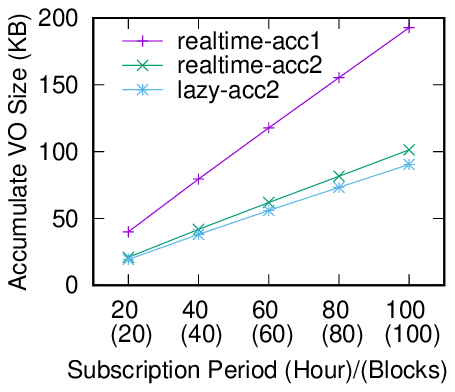}
    \end{subfigure}
    \caption{Subscription Query Performance (\textsf{WX})}\label{exp-fig:subwx}
\end{figure}

\begin{figure}[t]
    \centering
    \begin{subfigure}{.33\linewidth}
        \includegraphics[width=\linewidth]{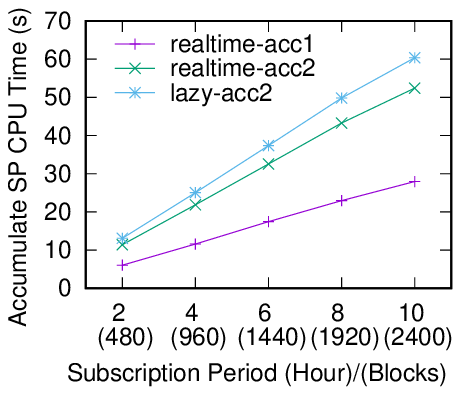}
    \end{subfigure}~%
    \begin{subfigure}{.33\linewidth}
        \includegraphics[width=\linewidth]{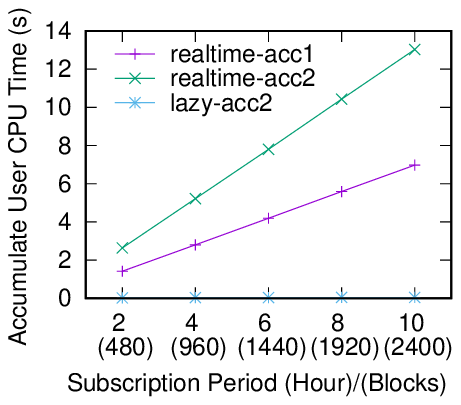}
    \end{subfigure}~%
    \begin{subfigure}{.33\linewidth}
        \includegraphics[width=\linewidth]{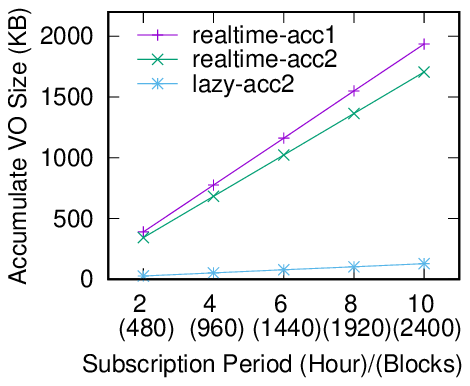}
    \end{subfigure}
    \caption{Subscription Query Performance (\textsf{ETH})}\label{exp-fig:subeth}
\end{figure}

\vspace{-2ex}
\subsection{Subscription Query Performance}
We next evaluate the performance for subscription queries. First, we examine the SP's query processing time with or without using the IP-Tree (denoted as $ip$ and $nip$) under the default setting with both intra- and inter-block indexes enabled. We randomly generate different numbers of queries. We set the default subscription period as 2 hours for \textsf{4SQ} and \textsf{ETH} and 20 hours for \textsf{WX}. As shown in Fig.~\ref{exp-fig:iptree}, the IP-Tree reduces the SP's overhead by at least 50\% in all cases tested. The performance gain in the \textsf{ETH} dataset (Fig.~\ref{exp-fig:iptree}(c)) is more substantial thanks to the sparser distribution of the data.

To compare real-time and lazy authentications, we consider two real-time schemes (with \emph{acc1} and \emph{acc2}) and one lazy scheme (with \emph{acc2} only, as \emph{acc1} does not support the aggregation of accumulative sets and proofs). We vary the subscription period from 2 hours to 10 hours for \textsf{4SQ} and \textsf{ETH} and 20 hours to 100 hours for \textsf{WX}.
 Figs.~\ref{exp-fig:sub4sq}--\ref{exp-fig:subeth} show the results of varying the subscription period.
 Clearly, the lazy scheme performs much better than the real-time schemes in terms of the user CPU time.
Furthermore, the CPU time and the VO size in the lazy scheme are increased only sub-linearly with increasing the subscription period.
This is because the lazy scheme can aggregate the proofs of mismatching objects across blocks. In contrast, the real-time schemes compute all the proofs immediately upon arrival of a new block, resulting in a worse performance. In terms of the SP CPU time, as the lazy scheme needs to sacrifice the SP's computation to aggregate mismatch proofs, its performance is generally worse than the real-time schemes when using the same accumulator.

\section{Conclusion}\label{sec:conclusion}
In this paper, we have studied, for the first time in the literature, the problem of verifiable query processing over blockchain databases. We proposed the vChain framework to ensure the integrity of Boolean range queries for lightweight users. We developed a novel accumulator-based ADS scheme that transforms numerical attributes into set-valued attributes and hence enables dynamic aggregation over arbitrary query attributes. Based on that, two data indexes, namely tree-based intra-block index and skip-list-based inter-block index, and one prefix-tree-based index for subscription queries were designed, along with a series of optimizations. While our proposed framework has been shown to be practically implementable, the robustness of the proposed techniques was substantiated by security analysis and empirical results.

This paper opens up a new direction for blockchain research. There are a number of interesting research problems that deserve further investigation, e.g., how to support more complex analytics queries; how to leverage modern hardware such as multi- and many-cores to scale performance; and how to address privacy concerns in query processing.
\clearpage

\noindent%
\textbf{Acknowledgments.} This work was supported by Research Grants Council of Hong Kong under GRF Projects 12201018, 12200817, 12244916, and CRF Project C1008-16G.

\printbibliography%

\appendices%

\section{Pseudo Codes of the IP-Tree Algorithms} \label{sec:ip_algorithms}

Algorithms~\ref{alg:invertedprefix-construct} and \ref{alg:subscription-query} respectively show the construction and query processing algorithms of the IP-tree index introduced in Section~\ref{sec:query_index}.

\begin{algorithm}[t]
    \caption{IP-Tree Construction (by the SP)}\label{alg:invertedprefix-construct}
    \small
    \SetKw{is}{is}
    \SetKw{not}{not}
    \SetKw{in}{in}
    \SetKw{and}{and}
    \SetKwFunction{BuildIPTree}{BuildIPTree}
    \Fn{\BuildIPTree{$\{query\}$}}{%
        \KwIn{All subscription queries $\{query\}$}
            $root.grid\gets FullRange$\;
            Create an empty queue $queue$\;
            $queue.enqueue(\langle root,\{query\} \rangle)$\;
            \While{$queue$ \is \not empty}{
                $\langle node, \{query_n\}\rangle \gets queue.dequeue()$\;
                \For{$q_i$ \in $\{query_n\}$}{
                    \uIf{$q_i$ full covers $node.grid$}{Add $\langle q_i,\textsf{full}\rangle$ to $node.RCIF$\;
                    }
                    \lElse{Add $\langle q_i,\textsf{partial}\rangle$ to $node.RCIF$}
                }

                Build $node.BCIF$ from fully cover queries\;
                \For{$sub\_g$~\in~\textsf{Split}($node.grid$)}{
                    $sub\_n.grid \gets sub\_g$\;
                    $sub\_q \gets \{q~|~node.RCIF[q]=\textsf{partial} \land intersect(q,sub\_g) \}$\;
                    $queue.enqueue(\langle sub\_n,\{sub\_q\} \rangle)$\;
                    $node.children.add(sub\_n)$\;
                }
            }
        }
\end{algorithm}

\begin{algorithm}[t]
    \caption{\small Subscription Query w. Intra-Index (by the SP)}\label{alg:subscription-query}
    \small
    \SetKw{is}{is}
    \SetKw{forin}{in}
    \SetKw{not}{not}
    \SetKw{and}{and}
    \SetKwFunction{QueryIntraNode}{QueryIntraNode}
    \SetKwFunction{SubscriptionIPTree}{SubscriptionIPTree}
    \Fn{\SubscriptionIPTree{$r_\textrm{IP},IntraRoot$}}{
        \KwIn{IP-Tree root $r_\textrm{IP}$, Intra-Index root $IntraRoot$}
        Create an empty queue $queue_1$\;
        $Q\gets \{\}$\;
        $queue_1.enqueue(\langle IntraRoot, Q\rangle)$\;
        \While{$queue_1$ \is not empty}
        {
            $\langle node, Q \rangle \gets queue_1.dequeue()$\;
            $Q\gets$ \QueryIntraNode{$r_\textrm{IP},node,Q$}\;
            \lFor{$n$ \forin $node.children$}{$queue_1.enqueue(\langle n, Q \rangle)$}
        }
    }

    \Fn{\QueryIntraNode{$r_\textrm{IP},n_{\textrm{intra}},Q$}}{
        \KwIn{IP-Tree root $r_\textrm{IP}$, Intra-Index node $n_{\textrm{intra}}$, processed mismatching queries $Q$}
        \KwOut{processed mismatching queries $Q$}
        Create an empty queue $queue_2$ and enqueue $r_\textrm{IP}$\;
        \While{$queue_2$ \is not empty}{
            $n\gets queue_2.dequeue()$\;
            $q_{f}\gets\{q | n.RCIF[q]=\text{full}\}\backslash Q$\;
            \For{$\langle \Upsilon,qs \rangle$ \forin $n.BCIF$}{
                \If{$n_{\textrm{intra}}.W \cap \Upsilon=\emptyset$}{
                    $\pi \gets \textsf{ProveDisjoint}(n_{\textrm{intra}}.W, \Upsilon, pk)$\;
                    Add $\langle \pi, \Upsilon \rangle$ to $q$.VO $\forall q \in (q_{f} \cap qs)\backslash Q$\;
                    $Q \gets Q \cup (q_{f} \cap qs)$\;

                }
            }
            \If{$n_{\textrm{intra}}$ \is leaf node}{
                Add $n_{\textrm{intra}}.o$ to $q.R$ for $q\in q_{f}\backslash Q$\;
            }
            $q' \gets \{\}$\;
            \For{$n'$ \forin $n.children$}{
                \lIf{$n'$ intersects $n_{\textrm{intra}}$}{$q' \gets q' \cup n'.queries$}
            }
            $q_{p}\gets \{q | n.RCIF[q]=\text{partial}\}\backslash q'$\;
            \For{$n'$ \forin $n.children$}{
                \lIf{$n'$ intersects $n_{\textrm{intra}}$}{
                    $queue_2.enqueu(n')$
                }\Else{
              $qs \gets q_{p} \cap n'.queries\backslash Q$\;
              $W' \gets \textsf{trans}(n'.grid)$\;
               $\pi \gets \textsf{ProveDisjoint}(n_{\textrm{intra}}.W, W', pk)$\;
              Add $\langle \pi, W' \rangle$ to $q$.VO $\forall q \in qs$\;
              $Q \gets Q \cup qs$;}

            }
        }
        \Return{$Q$}\;
    }
\end{algorithm}

\section{Proof of Theorem~\ref{THE:ACC-SEC}}\label{prof:acc-security}

\accsecuritytheorem*

\begin{proof}

We omit the proof of this theorem for the first construction of the multiset accumulator in  Section~\ref{sec:acc}, as it has been shown to hold under the $q$-SDH assumption in~\cite{papamanthou2011optimal}.

We prove this theorem holds for the second construction by contradiction. Support that there is an adversary who outputs two multisets $X_1$ and $X_2$, where $X_1 \cap X_2 \neq \emptyset$, and a valid set disjoint proof $\pi$. This means that
\begin{align*}
    e(d_A(X_1), d_B(X_2)) &= e(\pi, g) \Rightarrow \\
    e(g^{A(X_1)}, g^{B(X_2)}) &= e(\pi, g) \Rightarrow \\
    {e(g, g)}^{A(X_1)B(X_2)} &= e(\pi, g) \Rightarrow \\
    \pi &= g^{A(X_1)B(X_2)}
\end{align*}
On the other hand, because $X_1 \cap X_2 \neq \emptyset$, we can get $Cs^q \in A(X_1) B(X_2)$, where $C$ is a non-zero constant. That is, \linebreak $A(X_1) B(X_2) = Cs^q + Q(s)$, where $Q(s)$ is some polynomial without the $s^q$ term. Therefore, the adversary can get
\begin{align*}
    \pi &= g^{A(X_1)B(X_2)} = g^{Cs^q} \cdot g^{Q(s)} \Rightarrow \\
    g^{s^q} &= {(\pi / g^{Q(s)})}^{C^{-1}}
\end{align*}
which violates the $q$-DHE assumption. Therefore, by contradiction, the theorem holds.
\end{proof}

\section{Proof of Theorem~\ref{THE:ALG-SEC}}\label{prof:alg-security}

\algsecuritytheorem*

\begin{proof}

We prove this theorem by contradiction.

\underline{Case 1:} The result $R$ contains an object $o^*$ such that $o^*\notin~\{o_i\}$.

Recall that in the result verification procedure, the verifier will check the integrity of the object with respect to the $MerkleRoot$ stored in the blockchain. Therefore, a successful forge means either a collision of the underlying cryptographic hash function, or a break of the blockchain protocol, which yields a contradiction.

\underline{Case 2:} The result $R$ contains an object $o^*$ such that $o^*$ does not satisfy the query~$q$.

It is trivial to see that such a case is impossible, as the verifier will check it locally.

\underline{Case 3:} There exists an object $o_x$ in the query window or subscription period, which is not in $R$ but satisfies the query~$q$.

First, note that the verifier (running a light node) syncs block headers with the blockchain network. Thus, the verifier  always verifies the results with respect to the latest block header. Now suppose there is a missing object $o_x$.
During verification, the verifier will examine the multiset accumulative values which cover the whole query window or subscription period. The missing object $o_x$ must fall under one multiset accumulative value in the VO\@. As discussed in Section~\ref{sec:ads_query}, such a missing object implies that the attribute of this matching object intersects with the equivalent multiset of some part of the query condition. This means that the adversary is able to construct two multisets $X_1$ and $X_2$, such that $X_1 \cap X_2 \neq \emptyset$, along with a corresponding set disjoint proof, which contradicts to Theorem~\ref{THE:ACC-SEC}.
\end{proof}

\section{Supplemental Experiment Results}

This section presents some supplemental experiment results.

\subsection{Comparison with MHT}\label{sec:mht_exp}

\begin{figure}[t]
    \centering
    \begin{subfigure}{.33\linewidth}
        \includegraphics[width=\linewidth]{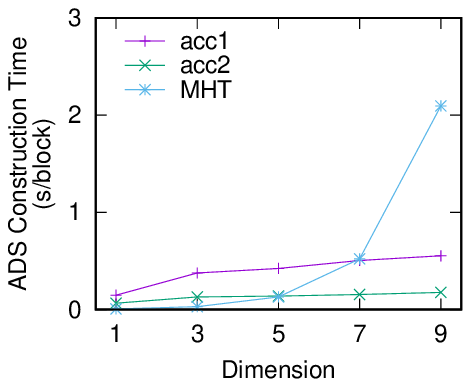}
    \end{subfigure}~%
    \begin{subfigure}{.33\linewidth}
        \includegraphics[width=\linewidth]{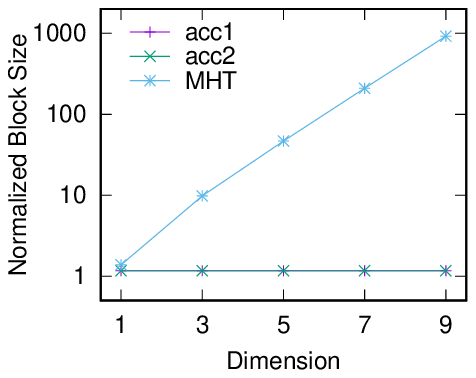}
    \end{subfigure}
    \caption{Comparison with MHT}\label{exp-fig:mht}
   \vspace{0.05in}
\end{figure}

As mentioned in Section~\ref{sec:basic_solution}, one major drawback of the traditional MHT-based solution is its prohibitively high overhead to support queries involving arbitrary attributes for multi-dimensional databases. To demonstrate that, we synthesize several datasets with different dimensionalities using the \textsf{WX} dataset and experimentally compare the MHT solution with our accumulator-based solutions in terms of the setup cost. Note that the original \textit{weather description} attribute is removed from the synthetic datasets, since MTHs cannot work with set-valued attributes.
As shown in Fig.~\ref{exp-fig:mht}(a), the construction time of our solutions is only slightly increased with dimensionality. In contrast, the MHT construction time is dramatically increased because it needs to build an MHT for every combination of attributes. Furthermore, to examine the ADS overhead, we show the average size of the ADS-embedded blocks in Fig.~\ref{exp-fig:mht}(b) (normalized by the original block size, plotted in log scale). While our solutions have a negligible fixed-size ADS regardless of data dimensionality, the space overhead incurred by MTH grows exponentially with dimensionality. In particular, when the dimensionality is higher than 3, the MHT's ADS  overhead is more than 10X--1,000X the original block size, which is unacceptable for practical use since the ADS would dominate the traffic and storage cost of a blockchain network.

\subsection{Impact of Selectivity}\label{sec:select_exp}

\begin{figure}[t]
    \centering
    \begin{subfigure}{.33\linewidth}
        \includegraphics[width=\linewidth]{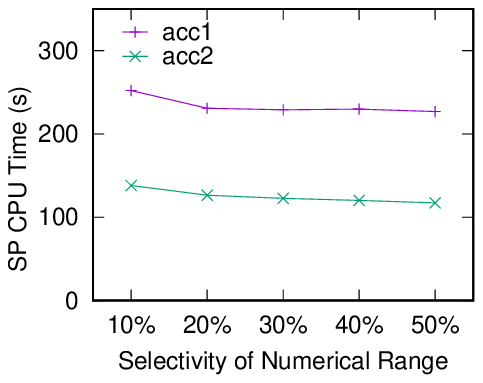}
    \end{subfigure}~%
    \begin{subfigure}{.33\linewidth}
        \includegraphics[width=\linewidth]{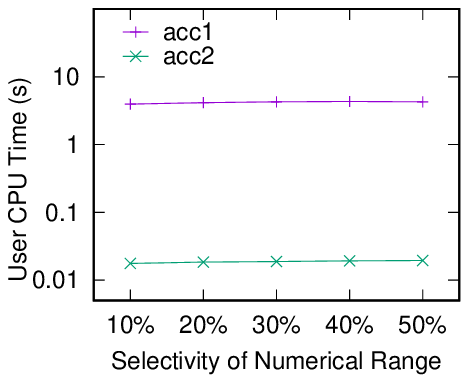}
    \end{subfigure}~%
    \begin{subfigure}{.33\linewidth}
        \includegraphics[width=\linewidth]{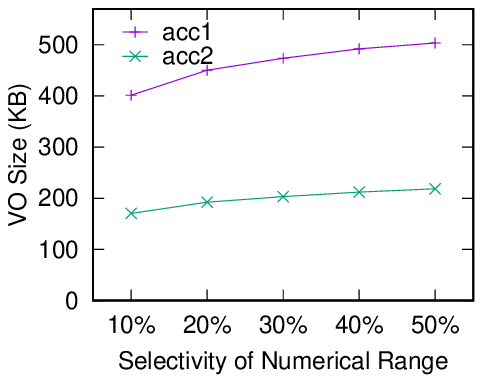}
    \end{subfigure}
    \caption{Impact of Selectivity (\textsf{4SQ})}\label{exp-fig:range4sq}
 \vspace{0.03in}
\end{figure}

\begin{figure}[t]
    \centering
    \begin{subfigure}{.33\linewidth}
        \includegraphics[width=\linewidth]{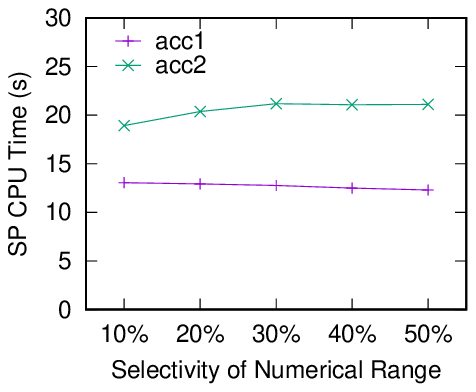}
    \end{subfigure}~%
    \begin{subfigure}{.33\linewidth}
        \includegraphics[width=\linewidth]{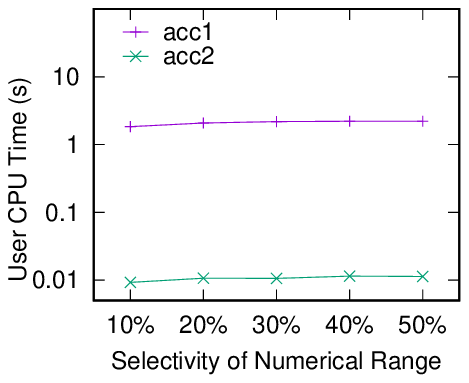}
    \end{subfigure}~%
    \begin{subfigure}{.33\linewidth}
        \includegraphics[width=\linewidth]{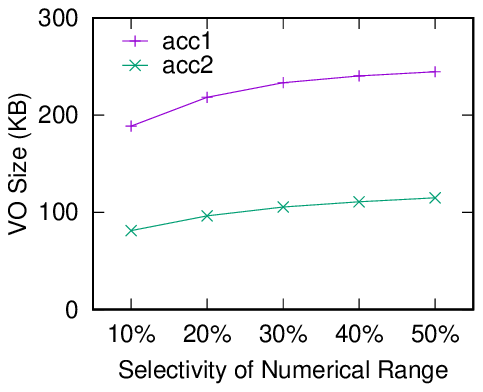}
    \end{subfigure}
    \caption{Impact of Selectivity (\textsf{WX})}\label{exp-fig:rangewx}
 \vspace{0.03in}
\end{figure}

\begin{figure}[t]
    \centering
    \begin{subfigure}{.33\linewidth}
        \includegraphics[width=\linewidth]{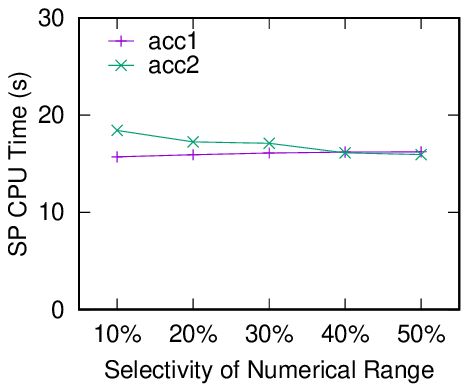}
    \end{subfigure}~%
    \begin{subfigure}{.33\linewidth}
        \includegraphics[width=\linewidth]{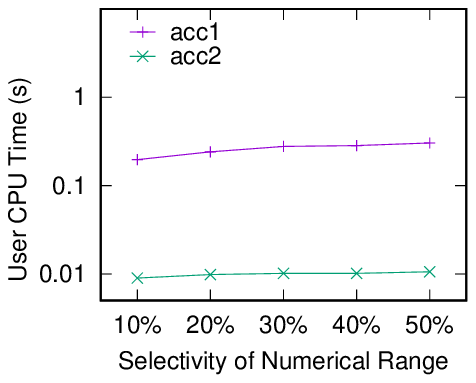}
    \end{subfigure}~%
    \begin{subfigure}{.33\linewidth}
        \includegraphics[width=\linewidth]{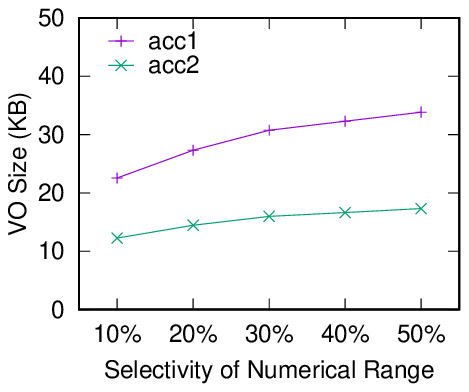}
    \end{subfigure}
    \caption{Impact of Selectivity (\textsf{ETH})}\label{exp-fig:rangeeth}
 \vspace{0.05in}
\end{figure}

Figs.~\ref{exp-fig:range4sq}--\ref{exp-fig:rangeeth} show the time-window query performance by varying the selectivity of the range predicate from 10\% to 50\%. The window size is fixed at 10 hours for \textsf{4SQ} and \textsf{ETH} and at 100 hours for \textsf{WX}.
Both the intra-block and inter-block indexes are enabled. For all datasets, the SP CPU time is generally decreased with increasing selectivity. This is because the SP query processing time is dominated by the proving of  mismatching objects. As a result, the more the objects selected, the less is the SP overhead. In contrast, the user CPU time remains largely the same under different settings. As for the VO size, it is slightly increased because of a larger number of hashes introduced by more query results. Overall, we can see that our solution performs efficiently under a wide range of selectivity settings.

\subsection{Impact of SkipList}\label{sec:skiplist_exp}
\begin{figure}[t]
    \centering
    \begin{subfigure}{.33\linewidth}
        \includegraphics[width=\linewidth]{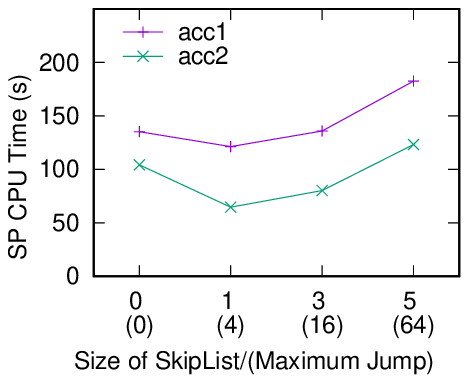}
    \end{subfigure}~%
    \begin{subfigure}{.33\linewidth}
        \includegraphics[width=\linewidth]{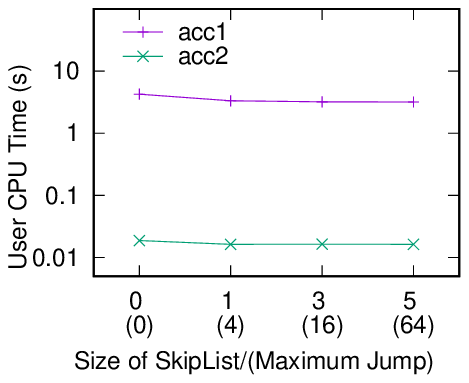}
    \end{subfigure}~%
    \begin{subfigure}{.33\linewidth}
        \includegraphics[width=\linewidth]{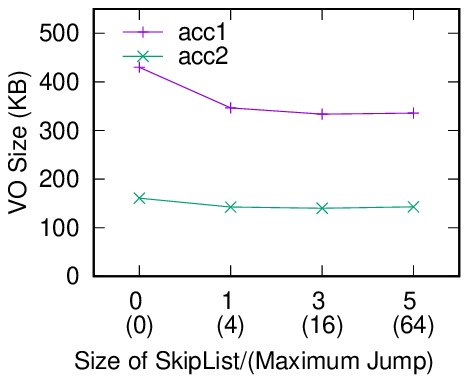}
    \end{subfigure}
    \caption{Impact of SkipList Size (\textsf{4SQ})}\label{exp-fig:skip4sq}
\end{figure}

\begin{figure}[t]
    \centering
    \begin{subfigure}{.33\linewidth}
        \includegraphics[width=\linewidth]{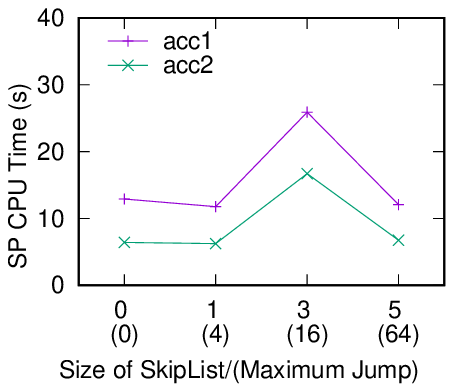}
    \end{subfigure}~%
    \begin{subfigure}{.33\linewidth}
        \includegraphics[width=\linewidth]{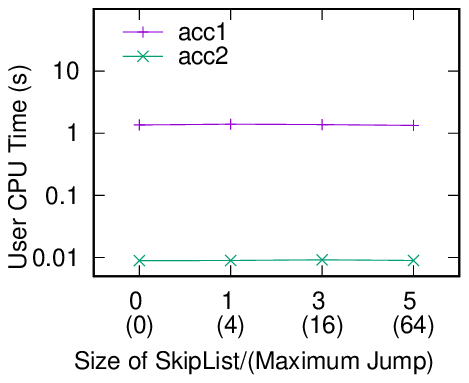}
    \end{subfigure}~%
    \begin{subfigure}{.33\linewidth}
        \includegraphics[width=\linewidth]{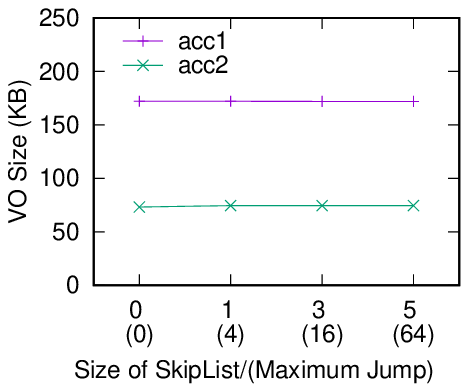}
    \end{subfigure}
    \caption{Impact of SkipList Size (\textsf{WX})}\label{exp-fig:skipwx}
\end{figure}

\begin{figure}[t]
    \centering
    \begin{subfigure}{.33\linewidth}
        \includegraphics[width=\linewidth]{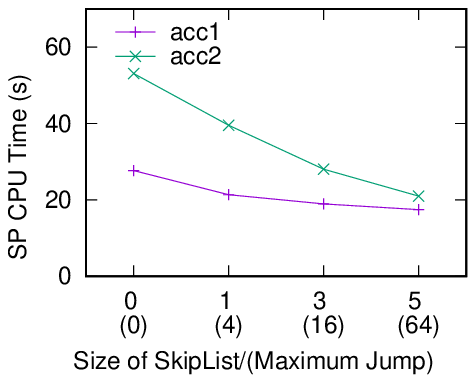}
    \end{subfigure}~%
    \begin{subfigure}{.33\linewidth}
        \includegraphics[width=\linewidth]{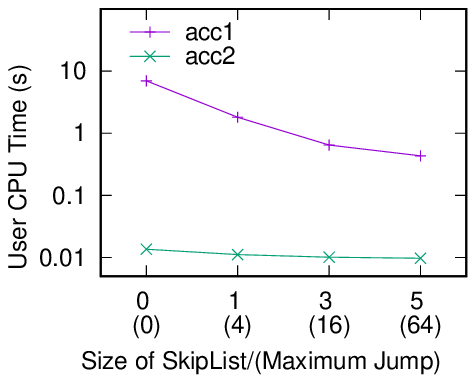}
    \end{subfigure}~%
    \begin{subfigure}{.33\linewidth}
        \includegraphics[width=\linewidth]{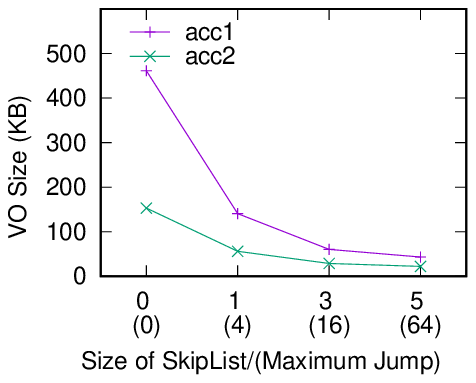}
    \end{subfigure}
    \caption{Impact of SkipList Size (\textsf{ETH})}\label{exp-fig:skipeth}
\end{figure}

This section investigates the impact of \emph{SkipList} in the inter-block index.
We set the query windows size the same as in Appendix~\ref{sec:select_exp}.
The size of \emph{SkipList} is varied from 0 to 5. Note that a size of zero means that we employ only an intra-block index but no inter-block index. For the other settings, both the intra-block and inter-block indexes are employed. Figs.~\ref{exp-fig:skip4sq}--\ref{exp-fig:skipeth} show the results for the three datasets, respectively. It is interesting to observe that the SP CPU time exhibits different trends for different datasets. This can be explained as follows. On the one hand, the \emph{SkipList} helps aggregate mismatch proofs across consecutive blocks. With an increased size of \emph{SkipList}, more mismatch proofs can be aggregated without accessing the actual blocks; hence the SP CPU time is reduced and a smaller VO is resulted. On the other hand, the larger the \emph{SkipList}, the more set elements will be added up as the input of the accumulator, which increases the SP CPU time. Furthermore, the effectiveness of using \emph{SkipList} for aggregating mismatch proofs depends on the distribution of the data. The combination of these factors contributes to the final SP CPU time. As a result, we observe fluctuations in SP CPU time for \textsf{4SQ} and \textsf{WX} but a steady decrease in \textsf{ETH} when the \emph{SkipList} size increases. Nevertheless, the user CPU time and the VO size are both monotonically reduced thanks to the aggregation of mismatching objects by the inter-block index.
Comparing $acc1$ and $acc2$, since \emph{acc2} can support online aggregation,
its VO size and also user CPU time are further reduced compared with those of \emph{acc1} in all cases tested.
\section{Practical Implementation}
\label{sec:prac_implement}

This section addresses the practical implementation issues of our proposed vChain framework. Since we need to compute an ADS and embed it into each block, the existing blockchains cannot be used directly. There are two possible solutions. First, we can develop a new chain by extending an open-source blockchain project.
Second, we can leverage smart contracts, trusted programs running on blockchains, to build a \emph{logical} chain that constructs and maintains the ADS for each block, on top of an existing blockchain (e.g., Ethereum~\cite{wood2014ethereum} or Hyperledger~\cite{IBM}). The advantage of the second solution is that we do not need to be bothered by the underlying system implementation, but focus on writing smart contracts to build the logical chain. Listing \ref{lst:buildvChain} shows an example of Ethereum smart contract that implements the vChain framework. Specifically, lines 4--14 define the structures of a block header and a block. The mapping structure, \emph{chainstorage}, indexes each block with the block hash (line 16). The function \emph{BuildvChain} first constructs the block header, including the intra-block index and the inter-block index (lines 22--24). Then, the block hash is computed (line 25) and, together with the input objects, assigned to the current block (lines 26--27). Finally, the newly constructed block is added to the \emph{chainstorage} using its block hash (line 28).

\begin{lstlisting}[language=Solidity,float,caption=Implementing vChain in Smart Contract, label=lst:buildvChain]
pragma solidity ^0.4.0;

contract vChainContract {
  struct BlockHeader {
    bytes32 PreBkHash;
    bytes32 MerkleRoot;
    bytes32 SkipListRoot;
  }
  struct Block {
    BlockHeader header;
    IntraIndex intraindex;
    InterIndex interindex;
    Object[] objects;
  }
  //stores each block
  mapping(bytes32=>Block) chainstorage;

  function BuildvChain(Object[] _objects,
                       bytes32 _PreBkHash) public {
    Block block;
    BlockHeader _header;
    _header.PreBkHash = _PreBkHash;
    (_header.MerkleRoot, block.intraindex) = BuildIntraIndex(_objects);
    (_header.SkipListRoot, block.interindex) = BuildInterIndex(chainstorage);
    bytes32 _currentBkHash = sha3(_header);
    block.objects = _objects;
    block.header = _header;
    chainstorage[_currentBkHash] = block;
  }
}
\end{lstlisting}

\end{document}